\documentclass[12pt]{article}
\usepackage[latin9]{inputenc}
\usepackage{geometry}
\usepackage{float}
\usepackage{mathtools}
\usepackage{amsmath}
\usepackage{amsthm}
\usepackage{amssymb}
\usepackage{setspace}
\usepackage{esint}
\usepackage[authoryear]{natbib}
\usepackage{url}
\makeatletter

\providecommand{\tabularnewline}{\\}
\floatstyle{ruled}
\newfloat{algorithm}{tbp}{loa}
\providecommand{\algorithmname}{Algorithm}
\floatname{algorithm}{\protect\algorithmname}

\theoremstyle{plain}

\usepackage{latexsym}
\usepackage{amsthm}\usepackage{amsfonts}\usepackage{graphicx}\usepackage{epsfig}
\usepackage{bm}
\newcommand*{\patchAmsMathEnvironmentForLineno}[1]{%
      \expandafter\let\csname old#1\expandafter\endcsname\csname #1\endcsname
      \expandafter\let\csname oldend#1\expandafter\endcsname\csname end#1\endcsname
      \renewenvironment{#1}%
         {\linenomath\csname old#1\endcsname}%
         {\csname oldend#1\endcsname\endlinenomath}}%
    \newcommand*{\patchBothAmsMathEnvironmentsForLineno}[1]{%
      \patchAmsMathEnvironmentForLineno{#1}%
      \patchAmsMathEnvironmentForLineno{#1*}}%
    \AtBeginDocument{%
    \patchBothAmsMathEnvironmentsForLineno{equation}%
    \patchBothAmsMathEnvironmentsForLineno{align}%
    \patchBothAmsMathEnvironmentsForLineno{flalign}%
    \patchBothAmsMathEnvironmentsForLineno{alignat}%
    \patchBothAmsMathEnvironmentsForLineno{gather}%
    \patchBothAmsMathEnvironmentsForLineno{multline}%
    }
\usepackage[mathlines,displaymath]{lineno}

\include{def}
\def\dispmuskip{\thinmuskip= 3mu plus 0mu minus 2mu \medmuskip=  4mu plus 2mu minus 2mu \thickmuskip=5mu plus 5mu minus 2mu}
\def\textmuskip{\thinmuskip= 0mu                    \medmuskip=  1mu plus 1mu minus 1mu \thickmuskip=2mu plus 3mu minus 1mu}
\def\beq{\dispmuskip\begin{equation}}    \def\eeq{\end{equation}\textmuskip}
\def\beqn{\dispmuskip\begin{displaymath}}\def\eeqn{\end{displaymath}\textmuskip}
\def\bea{\dispmuskip\begin{eqnarray}}    \def\eea{\end{eqnarray}\textmuskip}
\def\bean{\dispmuskip\begin{eqnarray*}}  \def\eean{\end{eqnarray*}\textmuskip}

\newcommand{\wh}{\hat}
\newcommand{\wt}{\widetilde}

\newcommand{\ov}{\overline}

\def\E{\mathbb{E}}                         

\def\s{\sigma}

\def\t{\theta}

\def\N{{\cal N}}

\def\CT{\text{\rm CT}}

\def\Var{\mathbb{V}}

\def\IS{\text{\rm IS}}

\def\Sup{\text{\rm Supp}}

\def\E{{\mathbb E}}                         
\def\V{{\mathbb V}}

\def\s{\sigma}

\def\t{\theta}

\def\N{{\cal N}}

\def\Var{\mathbb{V}}
\def\IS{\text{\rm IS}}

\def\E{\mathbb{E}}

\def\({\left(}
\def\){\right)}

\usepackage[mathlines,displaymath]{lineno}

\newtheorem{theorem}{Theorem}

\newtheorem{lemma}{Lemma}

\newtheorem{proposition}{Proposition}

\newtheorem{assumption}{Assumption}
\def \wt{\widetilde}
\def\wh{\widehat}
\makeatother

\providecommand{\theoremname}{Theorem}

\begin{document}

\title{Robustly Estimating the Marginal Likelihood for Cognitive Models via Importance Sampling}
\author{
M.-N. Tran\thanks{The University of Sydney Business School}
\and M. Scharth\footnotemark[1]
\and D. Gunawan\thanks{UNSW Business School, University of New South Wales}
\and R. Kohn\footnotemark[2]
\and S. D. Brown\thanks{School of Psychology, University of Newcastle}
\and G. E. Hawkins\footnotemark[3]
\footnote{The research of Tran, Gunawan, Kohn and Brown
was partially supported by Australian Research Council grant DP180102195. Hawkins was partially supported by Australian Research Council grant DE170100177. Tran and Kohn thank Michael Pitt for some useful discussions that led to the development of Proposition \ref{proposition:marg_lik}.}
}

\maketitle
\begin{abstract}
Recent advances in Markov chain Monte Carlo (MCMC) extend the scope of Bayesian inference to models for which the likelihood function is intractable. Although these developments allow us to estimate model parameters, other basic problems such as estimating the marginal likelihood, a fundamental tool in Bayesian model selection, remain challenging. This is an important scientific limitation because testing psychological hypotheses with hierarchical models has proven difficult with current model selection methods. We propose an efficient method for estimating the marginal likelihood for models where the likelihood is intractable, but can be estimated unbiasedly. It is based on first running a sampling method such as MCMC to obtain samples for the model parameters, and then using these samples to construct the proposal density in an importance sampling (IS) framework with an unbiased estimate of the likelihood. Our method has several attractive properties: it generates an unbiased estimate of the marginal likelihood, it is robust to the quality and target of the sampling method used to form the IS proposals, and it is computationally cheap to estimate the variance of the marginal likelihood estimator. We also obtain the convergence properties of the method and provide guidelines on maximizing computational efficiency. The method is illustrated in two challenging cases involving hierarchical models: identifying the form of individual differences in an applied choice scenario, and evaluating the best parameterization of a cognitive model in a speeded decision making context. Freely available code to implement the methods is provided. Extensions to posterior moment estimation and parallelization are also discussed.
\end{abstract}

\section{Introduction}

Many psychologically interesting research questions involve comparing competing theories: Does sleep deprivation cause attentional lapses? Does alcohol impair the speed of information processing or reduce cautiousness, or both? Does the forgetting curve follow a power or exponential function? In many cases, the competing theories can be represented as a set of quantitative models that are applied (``fitted'') to the observed data. We can then estimate a metric that quantifies the degree to which each model accounts for the patterns observed in data balanced against its flexibility. Model flexibility is often defined as the range of data patterns that a model can predict, which includes patterns that were observed as well as patterns that were not observed. Flexibility is an important consideration in model choice as it assesses the degree to which a theory is suitably constrained to predict the observed data and that it does not simply predict many possible patterns that were not observed; the latter form of model is theoretically non-informative as it ``predicts'' almost any pattern that could be observed. Many methods exist that attempt to quantitatively account for the flexibility of models, and hence inform model choice, including likelihood ratio tests, various information criteria (e.g., Akaike, Bayesian and Deviance Information Criteria; AIC, BIC, and DIC, respectively), minimum description length, cross validation and others, with varying degrees of success. The emerging gold standard in quantitative psychology is the marginal likelihood; the commonly cited Bayes factor is the ratio of the marginal likelihoods of two models.

The marginal likelihood of a model is the average likelihood of the model given the data, where the average is taken over the prior distribution of the model parameters. By integrating the likelihood across the prior predictive space of a model, the marginal likelihood has the attractive property that it inherently accounts for the flexibility of the model. This removes the need for post-hoc complexity corrections to goodness of fit measures, as implemented in the ``information criteria'' metrics (e.g., AIC, BIC, DIC), because it provides a guarantee that a model which accounts for the observed patterns in data and few other possible though unobserved patterns will be quantitatively favored over a competing model which accounts for the observed data but also provides the capacity to account for many other patterns that were not observed. This occurs because the latter model would have low likelihood across the (many) regions of the prior space where data were not observed, which lowers the marginal likelihood.

Although theoretically attractive, estimating the marginal likelihood for psychologically interesting models poses a number of practical challenges. 
The likelihood function is the density function of the data with the participant-level parameters (also called random effects) integrated out, when viewed as a function of the group-level parameters. It is usually analytically or computationally intractable for psychological and statistical models because it is an integral that cannot be evaluated. 
This can be the case even in conceptually very simple models. For example, in generalized linear mixed models where random effects (e.g., model parameters for individual participants) account for the dependence between observations from the same individual \citep{Fitzmaurice:2011}, the likelihood is often intractable because it is an integral over the random effects. Recent advances in particle Markov chain Monte Carlo (P-MCMC) methods extend the scope of the applications of Bayesian inference to cases where the likelihood is intractable \citep[see][]{Andrieu:2009,Andrieu:2010,Chopin2013,Tran:2016}. However, some basic problems such as estimating the marginal likelihood and the standard error of this estimator remain challenging, because the marginal likelihood is the density of the data with both the group-level parameters and random effects integrated out. Section \ref{Sec:marginal likelihood in hierarchical models} discusses some existing approaches for addressing this challenge.

Our article proposes an importance sampling (IS) approach for
estimating the marginal likelihood  when the likelihood is intractable but can be estimated unbiasedly.
The method is implemented by first running a sampling scheme such as MCMC, whose draws are used to form an importance distribution for the fixed parameters.  We then estimate the marginal likelihood using an unbiased estimate of the likelihood combined with this importance distribution.
We call this approach Importance Sampling Squared $\left(\textrm{IS}^{2}\right)$, as it is itself an importance sampling (IS) procedure and we can often estimate the likelihood by IS. We claim that it has several advantages over competing methods for estimating the marginal likelihood. The three most important ones are that a) it produces an unbiased estimate of the marginal likelihood; b)  the  method is robust against potential issues with the samples used to form the proposals; and c) it provides an easily computable  and cheap estimate of the variability of the marginal likelihood estimator. Unbiasedness means that
we can approach the true value and easily lower the variance of the estimator by getting more samples. Robustness means that the marginal likelihood estimator is simulation consistent even when the MCMC sampler has not yet converged to the posterior distribution of interest. This can mean that the MCMC may have a slightly perturbed posterior as a target or that the MCMC has the posterior as a target but it may not as yet have converged.
The reason that $\left(\textrm{IS}^{2}\right)$ works robustly is that it uses importance weights to correct for bad importance samples.
Estimating the variability of the estimator directly means that we do not have to replicate the algorithm to obtain reliable estimates of the standard error of the estimator. These observations in turn suggest that it may be desirable to run the initial
MCMC in parallel, and without rigorously requiring each MCMC to converge before forming the proposals for IS. Sections
\ref{Sec:marginal likelihood in hierarchical models}, \ref{sec:Conclusions} and Section~\ref{Appendix:IS2-method} of the appendix expand on these points.

The rest of the paper is organized as follows. Section~\ref{Sec:marginal likelihood in hierarchical models} introduces the general
model under consideration, reviews several
competing approaches for estimating the marginal likelihood, and briefly compares them to our approach. Section~\ref{Sec:IS2algorithm} gives  a detailed
description of the $\textrm{IS}^{2}$ method.  Section~\ref{sec:real applications} illustrates the $\textrm{IS}^{2}$ method in two applications, which
show that the method gives accurate estimates of the marginal likelihood and its standard error in reasonable computing time. Section~\ref{sec:Conclusions} concludes and
discusses future work on speeding up the computation. Sections~\ref{Appendix:IS2-method} to~\ref{Proofstheorem} are five appendices. Section~\ref{Appendix:IS2-method} outlines how  our results on estimating the marginal likelihood extend to estimating
posterior expectations. Section~\ref{Appendix:IS2-effect-on-N} studies the effect on IS of using an estimated likelihood.
 Section~\ref{Appendix:estimation-details} provides further estimation
 details for the two applications in Section~\ref{sec:real applications} and Section~\ref{Appendix:additional-LBA-applications} provides additional
 applications of the method. Finally, Section~\ref{Proofstheorem} contains the proofs of all the results in the paper.

\section{Marginal Likelihood Estimation for Hierarchical Models}\label{Sec:marginal likelihood in hierarchical models}

This section defines the class of hierarchical models for which the density function of the observations (individual participant data),
conditional on the group-level parameters ${\theta}$ and a vector of individual level parameters ${\alpha}$ (random effects) is available, but the likelihood is intractable because it is an integral that cannot be computed over the individual random effects. We also briefly review competing approaches.

Let ${y}_{j}$ be the vector of observations (responses) for the $j^{th}$ subject and define ${y}={y}_{1:S}=\left({y}_{1},...,{y}_{S}\right)$ as the vector of observations for all $S$ subjects.
Let ${\alpha}_{j}\in{\mathbb{R}^{d_{\alpha}}}$, where $\mathbb{R}^d$ denotes $d$-dimensional Euclidean space, be the vector of random effects (subject-level parameters) for subject $j$ and define $p\left({\alpha}_{j}|{\theta}\right)$ as its density.
We define ${\alpha}={\alpha}_{1:S}=\left({\alpha}_{1},...,{\alpha}_{S}\right)$ as the vector of all random effects in the model (i.e., all subject-level parameters).
Let ${\theta}\in{\mathbb{R}^{d_{\theta}}}$ be the vector of unknown group-level parameters and let $p\left({\theta}\right)$ be the prior for ${\theta}$.
Assuming that the $y_j, \alpha_j$, $j=1,...,S$ are independent given $\theta$, the joint density of the random effects and the observations is
\begin{equation}
p\left({y},{\alpha}|{\theta}\right)=\prod_{j=1}^{S}p\left({y}_{j},{\alpha}_{j}|{\theta}\right)=\prod_{j=1}^{S}p\left({y}_{j}|{\alpha}_{j},{\theta}\right)p\left({\alpha}_{j}|{\theta}\right).
\end{equation}
We assume the densities $p\left({y}|{\alpha},{\theta}\right)$, $p\left({\alpha}|{\theta}\right)$, and $p\left({\theta}\right)$ are available analytically or numerically (e.g., by numerical integration).
However, even when these densities are separately available, the likelihood of the hierarchical model is typically analytically intractable because it is an integral over the random effects:
\begin{equation}
p\left({y}|{\theta}\right)=\prod_{j=1}^{S}p\left({y}_{j}|{\theta}\right)\;\;\textrm{with}\;\;p\left({y}_{j}|{\theta}\right)=\int p\left({y}_{j}|{\alpha}_{j},{\theta}\right)p\left({\alpha}_{j}|{\theta}\right)d{\alpha}_{j}.\label{eq:likelihood}
\end{equation}
By Bayes' rule, we can express the joint posterior density of
${\theta}$ and ${\alpha}$ as
\begin{equation}
\pi\left({\theta},{\alpha}\right)\coloneqq p\left({y}|{\alpha},{\theta}\right)p\left({\alpha}|{\theta}\right)p\left({\theta}\right)/p\left({y}\right),
\end{equation}
where
\begin{align} \label{eq: marglik 1}
p\left({y}\right) & =\int\int p\left({y}|{\alpha},{\theta}\right)p\left({\alpha}|{\theta}\right)p\left({\theta}\right)d{\alpha}d{\theta}
\end{align}
is the marginal likelihood.

The main goal of the article is to develop an efficient method for estimating the marginal likelihood, given samples from the posterior distribution over the group-level parameters and individual-level random effects. These samples will have been obtained through generic sampling routines (e.g., JAGS, STAN) or custom \citep[e.g., differential evolution-MCMC (DE-MCMC) in][]{Turner2013}. Even custom samplers still suffer from the problems of high autocorrelation between successive iterates of samples from the posterior, and of slow or uncertain convergence in some problems; for example, when the vector of random effects is high dimensional. These problems can lead to unreliable or biased posterior samples, in which case model selection metrics derived from those posterior samples can also be dramatically wrong. This issue has the potential to influence the reliability of some modern developments in estimating the marginal likelihood for cognitive models, including bridge sampling \citep{gronau2017tutorial,gronau2019simple} and thermodynamic integration \citep{evans2019thermodynamic}. Both of those approaches have the same structure: they begin with posterior samples, and post-process these samples to estimate the marginal likelihood. Both bridge sampling and thermodynamic integration are therefore sensitive to the accuracy of the sampling method -- when the posterior samples are biased in some way, the estimated marginal likelihood will also be biased. An alternative approach is to bypass posterior samples altogether. For example, \citet{evans2018bayes} proposed estimating the marginal likelihood directly by Monte Carlo integration over the prior. This method works well in small problems, but it can be highly inefficient. This ``direct'' method quickly becomes computationally infeasible in high-dimensional problems, or with hierarchically-specified models. 


Our paper proposes a new approach. The basic idea is that we first obtain  samples of the model parameters (${\theta}$; group level), but not the random effects (${\alpha}$; individual participants). These samples are used to form efficient proposal densities for an importance sampling algorithm on the model parameters. The IS$^2$ method therefore does not suffer from the same potential drawbacks as the bridge sampling and thermodynamic integration methods, because the  samples are only used to form the proposals for $\theta$.  This means that unreliable samples of $\theta$ may lead to inefficient proposal densities, which can decrease the efficiency of the IS$^2$ method, but will not lead to bias in the marginal likelihood. In fact, the IS$^2$ method makes it possible to obtain unbiased and simulation consistent estimators of the marginal likelihood without ensuring that the underlying Markov chain(s) have in fact converged,
although they need to produce reasonable estimates of the posterior. Section~\ref{sec:Conclusions} discusses the robustness property of the IS$^2$ estimator and how to use it to speed up the estimation of the marginal likelihood and posterior moments.

\section{Estimating the Marginal Likelihood by $\textrm{IS}^{2}$  \label{Sec:IS2algorithm}}

The likelihood of hierarchical models of the form in Equation \eqref{eq:likelihood} can be analytically intractable, though it can be estimated unbiasedly using importance sampling (other hierarchical models such as state space models may require a particle filter).
Let $\left\{ m_{j}\left({\alpha}_{j}|{\theta},{y}_{j}\right);j=1,...,S\right\} $
be a family of proposal densities that we use to approximate the conditional
posterior densities $\left\{ \pi\left({\alpha}_{j}|{\theta}\right);j=1,...,S\right\}$. Let $g_{\IS}\left({\theta}\right)$ be a proposal density for the group-level parameters, $\theta$.
We note that many different methods can be used to obtain efficient and reliable proposals $g_{\IS}\left({\theta}\right)$ for the group-level parameters and  $m_{j}\left({\alpha}_{j}|{\theta},{y}_{j}\right)$ for the random effects $\alpha_j$,
for each subject $j=1,...,S$. 
A simple frequently used approach takes the proposal as a multivariate Student t distribution with a small number of degrees of freedom, whose location is a mode of the log-likelihood and whose scale matrix is the inverse of minus the Hessian matrix at this mode. Our article constructs the proposal distribution by first running MCMC, even just a short MCMC, to obtain samples from the posterior distribution. We then construct a proposal distribution in the IS$^2$ procedure by fitting a mixture of normal or Student t distributions to these samples, as outlined in Section \ref{sec:real applications} and Appendix \ref{Appendix:estimation-details}.

The density $p\left({y}_{j}|{\theta}\right)$ is estimated unbiasedly by
\begin{equation}\label{eq: importance weights}
\widehat{p}\left({y}_{j}|{\theta}\right)=\frac{1}{N}\sum_{i=1}^{N}w\left({\alpha}_{j}^{\left(i\right)},{\theta}\right),\;\;\textrm{where}
\;\;w\left({\alpha}_{j}^{\left(i\right)},{\theta}\right)=\frac{p\left({y}_{j}|{\alpha}_{j}^{\left(i\right)},{\theta}\right)
p\left({\alpha}_{j}^{\left(i\right)}|{\theta}\right)}{m_{j}\left({\alpha}_{j}^{\left(i\right)}|{\theta},{y}_{j}\right)},
\end{equation}
where ${\alpha}_{j}^{\left(i\right)}\overset{iid}{\sim}m_{j}\left({\alpha}_{j}|{\theta},{y}_{j}\right)$ and $N$ is the number of importance samples used in the likelihood estimate, which we will refer to as the number of ``particles''. We note that for importance sampling to be effective, it is necessary for the support of $m_{j}\left({\alpha}_{j}|{\theta},{y}_{j}\right)$ to contain the support of its corresponding conditional density (i.e., the tails of the proposal should be fatter than the tails of the target). This
condition for importance sampling is usually easily satisfied using the defensive sampling approach in \cite{Hesterberg1995}. Appendix \ref{Appendix:estimation-details} discusses the defensive sampling approach and the proposal $m_{j}\left({\alpha}_{j}|{\theta},{y}_{j}\right)$. Hence,
\begin{equation}\label{eq: likelihood estimator}
\widehat{p}_N\left({y}|{\theta}\right)=\prod_{j=1}^{S}\widehat{p}\left({y}_{j}|{\theta}\right)
\end{equation}
is an unbiased estimator of the likelihood $p\left({y}|{\theta}\right)$.
To choose an optimal number of particles  $N$, discussed in Section~\ref{Subsec:tradeoffIS},
we use the delta method to estimate the variance 
\begin{equation}
\V\left(\log\widehat{p}_N\left({y}|{\theta}\right)\right)=\sum_{j=1}^{S}\V\left(\log\widehat{p}\left({y}_{j}|{\theta}\right)\right)
\label{varianceeqn}
\end{equation}
of $\log\widehat{p}_N\left({y}|{\theta}\right)$. This gives the estimate
\[
\wh \V\left(\log\widehat{p}\left({y}_{j}|{\theta}\right)\right)\approx\frac{\sum_{i=1}^{N}w\left({\alpha}_{j}^{\left(i\right)},{\theta}\right)^{2}}{\left(\sum_{i=1}^{N}w\left({\alpha}_{j}^{\left(i\right)},{\theta}\right)\right)^{2}}-\frac{1}{N}.
\]

We now rewrite Equation~\eqref{eq: marglik 1} to show how it can be estimated by IS. 
Let $u$ consist of all random variables used 
to construct $\wh p_N(y|\theta)$ for a given value of $\theta$, with $p_N(u|\theta)$ the density of $u$. 
In practice, as will be clear shortly, it is unnecessary to know $p_N(u|\theta)$. We also equivalently write 
$\wh p_N(y|\theta)$ as $\wh p_N(y|\theta,u)$. Then $\wh p_N(y|\theta)$ is unbiased
if 
\begin{align*}
p(y|\theta) & = \int_U \wh p_N(y|\theta,u) p_N(u|\theta) du . 
\end{align*}
Let $g_{\IS}(\theta)$ be the proposal for $\theta$ obtained by MCMC or otherwise. Then, the marginal likelihood is given by 
\begin{align} \label{eq: marg lik II}
p(y) & = \int_\Theta  \int_U \wh p_N(y|\theta,u) p_N(u|\theta) p(\theta) du d\theta ,\notag \\
& = \int_\Theta  \int_U   \frac{\wh p_N(y|\theta,u) p_N(u|\theta)p(\theta)}    {g_{\IS}(\theta)p_N(u|\theta)} g_{\IS}(\theta)p_N(u|\theta) du d\theta,\notag \\
& = \int_\Theta \int_U  w(\theta, u)    g_{IS}(\theta)  p_N(u|\theta) du d\theta , 
\quad \text{where}\quad 
w(\theta,u)  =  \frac{\wh p_N(y|\theta,u)p(\theta) }    {  g_{\IS}(\theta)}  .
\end{align}
This leads to the IS$^2$ estimator, 

\begin{align} \label{marg lik est }
\wh p_{\IS^2}(y) & =\frac1M \sum_{i=1}^M \wt w(\theta_i) 
\quad \text{where}\quad 
\widetilde{w}\left({\theta}_{i}\right)=w(\theta_i,u_i),\;\;i=1,...,M . 
\end{align}

Algorithm \ref{alg:Importance-Sampling-Squared} outlines the $\textrm{IS}^{2}$ algorithm for obtaining an estimate of the marginal likelihood. On a practical note, the algorithm is well-suited to efficient parallelization, by processing particles (importance samples) independently in steps (1) and (2).

\begin{algorithm}[h]\label{IS2_Algorithm}
\caption{$\textrm{IS}^{2}$ algorithm for estimating the marginal likelihood \label{alg:Importance-Sampling-Squared}}
\begin{enumerate}
\item Generate ${\theta}_{i}\overset{iid}{\sim}g_{\IS}\left({\theta}\right)$
and compute the likelihood estimate $\widehat{p}_N\left({y}|{\theta}_{i}\right)$, $i=1,...,M$ via Equations~\eqref{eq: importance weights} and \eqref{eq: likelihood estimator}.
\item Compute the weights
\begin{equation}\label{eq: IS2 weights}
\widetilde{w}\left({\theta}_{i}\right)=\frac{\widehat{p}_N\left({y}|{\theta}_{i}\right)p\left({\theta}_{i}\right)}{g_{\IS}\left({\theta_i}\right)},\;\;i=1,...,M
\end{equation}

\item The $\textrm{IS}^{2}$ estimator of the marginal likelihood $p\left({y}\right)$
\begin{equation}
\widehat{p}_{\IS^{2}}\left({y}\right)=\frac{1}{M}\sum_{i=1}^{M}\widetilde{w}\left({\theta}_{i}\right).
\end{equation}

\end{enumerate}
\end{algorithm}

The usual way to estimate the standard error of the marginal likelihood estimator is by replication, i.e.,  by estimating the marginal likelihood estimator a number of times and then computing the standard deviation of these estimates. However, reliable estimation by replication can take a long time. 
An important advantage of our method is that it  is straightforward to estimate the standard error of the marginal likelihood estimator $\widehat{p}_{\IS^{2}}\left({y}\right)$. This affords researchers greater confidence in application of the estimated marginal likelihood, and also permits a simpler investigation of the tradeoff between efficiency and bias. The variance estimator is
\begin{equation}\label{eq:variance of logllh estimator}
\widehat{\V}\left(\widehat{p}_{\IS^{2}}\left({y}\right)\right)=\frac{1}{M}\widehat{\sigma}^2_{p_{\IS^2}(y)},\;\;\text{where}\;\; \widehat{\sigma}^2_{p_{\IS^2}(y)}=\frac{1}{M}\sum_{i=1}^{M}\left(\widetilde{w}\left({\theta}_{i}\right)-\widehat{p}_{\IS^{2}}\left({y}\right)\right)^{2}.
\end{equation}

It is clear that under mild conditions the IS$^2$ estimator is unbiased, simulation consistent and tends to normality as it is an IS estimator. 
Theorem~\ref{lem: robust} formally states these results and Section \ref{Proofstheorem} of the Appendix gives its proof. 



\begin{theorem}\label{lem: robust}
Let $M$ be the number of samples for ${\theta}$ and $N$ the number of particles for estimating
the likelihood. Under the assumptions in Theorem \ref{the:IS1} in Section~\ref{SS: technical results Sec 6} of the Appendix.

(i) $\E\left(\widehat{p}_{\IS^{2}}\left({y}\right)\right)=p\left({y}\right)$
and $\widehat{p}_{\IS^{2}}\left({y}\right)\overset{a.s.}{\rightarrow}p\left({y}\right)$
as $M\rightarrow\infty$ for any $N\geq1$.


(ii) $\sqrt{M}\big(\widehat{p}_{\IS^{2}}\left({y}\right)-p\left({y}\right)\big)\overset{d}{\rightarrow}N\left(0,{\sigma}^2_{p_{\IS^2}(y)}\right)$
and $\widehat{\sigma}^2_{p_{\IS^2}(y)}\overset{a.s.}{\rightarrow}{\sigma}^2_{p_{\IS^2}(y)}=\V(\widetilde{w}\left({\theta}\right))$
as $M\rightarrow\infty$ for any $N\geq1$.
\end{theorem}

There is a trade-off between the statistical precision of the IS$^2$ estimator (in terms of the variance 
from Equation~\eqref{eq:variance of logllh estimator}) and the computational cost of obtaining the estimate. A large number of particles $N$ gives a more accurate estimate of the likelihood with greater computational cost, while a small $N$ results in a likelihood estimator $\widehat p_N(y|\theta)$ that is cheap to evaluate but has larger variance. Section~\ref{Appendix:IS2-effect-on-N} of the Appendix provides theoretical and practical guidelines
of this tradeoff. Proposition \ref{proposition:marg_lik} in 
Section \ref{SS: optimal N marg likel} shows that the optimal $N$ in the tradeoff between accuracy and computational cost is such that the variance of the log-likelihood estimator $\V(\log\widehat p(y|\theta))$ is approximately 1.

\section{Applying the IS$^2$ Method to Data}\label{sec:real applications}

This section describes two applications of the IS$^2$ method to real data. The applications are chosen from quite
different domains -- the first focuses on choice in a health context, the second focuses on an experimental
 manipulation in cognition. We chose these applications as they are challenging cases for estimating marginal likelihoods,
 and thus represent good test cases for illustrating the power of our method. They also emphasize that the IS$^2$ method is not restricted to a particular domain of application or model class. It is a general method for estimating the marginal likelihood and the standard error of that estimate and can be applied in a variety of other contexts where quantitative models are estimated from data in a hierarchical framework including the study of attention, decision making, categorization, and memory. 

\subsection{Individual differences in health-related decisions}\label{sec:gmnlmodel}

The first application explores individual differences in preferences for health decision making, by examining choices for hypothetical appointments with a health practitioner. A common approach to understanding individual differences in applied contexts is to assume that different individuals have different \emph{utility} parameters, where the utility represents the subjective value placed on the \emph{attributes} that describe the products and services on offer. In this sense the utility is a quantitative estimate of what people like and what they dislike. A common way to model individual differences in choices is to then assume individual differences in \emph{utility} parameters, commonly described as ``taste heterogeneity''. The mixed logit model is the most common framework for modeling individual differences in utilities \citep{train2009discrete}, which assumes utility parameters are multivariate normal distributed in the population. Nevertheless, more recent research suggests that all individuals might hold the same (or have sufficiently similar) utilities except that those utilities can be scaled up or down for different individuals \citep{fklw2010}. This approach assumes that the ratio of the utilities for pairs of items is constant across participants though the absolute difference between those utilities can differ. This ``scale heterogeneity'' means that some people make more random choices than others, though their latent preferences, in the limit of very many trials, are similar.

The key model comparison question in this context is whether the choices observed in a population of participants are more consistent with individual differences in \emph{latent preferences} (taste) or individual differences in \emph{choice consistency} (scale). We tested this idea for the data reported in \cite{fklw2010} in a study where $S=79$ women made choices about $T=32$ choice scenarios for a pap smear test, where each scenario was defined by a combination of the values of the $K=5$ choice attributes listed in Table \ref{tab:choices} (not all unique combinations were presented). 
For example, one choice scenario might be that the pap smear would be conducted by a female doctor who is unknown to the patient, the test is not due though the doctor recommends it, with a cost of \$20. The participant is then asked whether they would take this test or not.

\begin{table}[h]
\caption{Choice attributes for the pap smear data set and their associated parameters.}\label{tab:choices}
\begin{center}
\begin{tabular}{lll}
\hline\hline
Choice attributes & Values & Parameters \\
\hline
Alternative specific constant for ``take the test'' & 1 & $\beta_0, \sigma_0$ \\
Whether patient knows doctor &  0 (no), 1 (yes) & $\beta_1, \sigma_1$ \\
Whether doctor is male &  0 (no), 1 (yes) & $\beta_2, \sigma_2$ \\
Whether test is due &  0 (no), 1 (yes) & $\beta_3, \sigma_3$ \\
Whether doctor recommends test &  0 (no), 1 (yes) & $\beta_4, \sigma_4$ \\
Test cost & \{0, 10, 20, 30\} & $\beta_5$ \\
\hline\hline
\end{tabular}
\end{center}
\end{table}

Our goal is to test whether choices in the pap smear test data provide greater support for the presence of individual differences in latent preferences, or the presence of individual differences in latent preferences \emph{and} choice consistency. For this goal, we estimated different models which instantiated the different hypotheses. We then used the IS$^2$ method to estimate the marginal likelihoods of those models, and used those estimates for model selection via Bayes factors.

\subsubsection{Behavioural model}

We analyse the choice data with the generalized multinomial logit (GMNL) model \citep{fklw2010}, a generalization of the mixed logit model, that accounts for individual differences in both latent preferences and choice consistency. We let the observed choice for participant $j$ in scenario $i$ be $y_{ij}=1$ if the participant chooses to take the test and $y_{ij}=0$ otherwise.
The general form of the GMNL for the probability of individual $j$ on trial $i$ selecting the test is
\begin{equation}\label{eq:choiceprob}
p(y_{ij}=1|X_{ij},\beta^{(j)})=\frac{\exp(\beta_{0j}+\sum_{k=1}^{K}\beta_{kj}x_{ijk})}{1+\exp(\beta_{0j}+\sum_{k=1}^{K}\beta_{kj}x_{ijk})},
\end{equation}
where $\beta^{(j)}=(\beta_{0j},\beta_{1j},\ldots,\beta_{Kj})'$ and $X_{ij}=(x_{ij1},\ldots,x_{ijK})'$ are the vectors of utility weights and choice attributes respectively. As is standard in most consumer decision research, the GMNL model assumes that the attributes comprising an option (i.e., a hypothetical appointment) each have an associated utility, and the utility of an option as a whole is the sum of the component attribute utilities. The summed utility probabilistically generates a choice via the Luce choice rule, as shown in Equation \eqref{eq:choiceprob}.
The model also allows for an {\it a priori} bias toward accepting or rejecting the test at each trial independent of the feature values of the test, known as an alternative specific constant; $\beta_{0j}=\beta_{0}+\eta_{0j}$, with $\eta_{0j}\sim \N(0,\sigma_{0}^2)$.

The utilities for individual participants are modeled as
\beqn
\beta_{jk}=\lambda_j\beta_k+\gamma\eta_{jk}+(1-\gamma)\lambda_j\eta_{jk},\qquad \lambda_j=\exp(-\delta^2/2+\delta\zeta_{j}), \qquad k=1,\ldots,K,
\eeqn
with $\eta_{jk}\sim \N(0,\sigma_k^2)$ and  $\zeta_{j}\sim \N(0,1)$.
The expected value of the scaling coefficients $\lambda_j$ is one, implying that $\E(\beta_{jk})=\beta_k$. We set the across-participant variance parameter for the cost attribute to zero ($\sigma_5^2=0$) as we found no evidence of heterogeneity across individuals for this attribute, beyond the scaling effect. 
As the sum of the probabilities over the possible choice options in Equation \eqref{eq:choiceprob} is 1 (i.e., take the test, or not), to make the model identified we set the coefficients $\beta$ with respect to $y_{ij}=0$ (not taking the test) to zero.

The GMNL model accounts for individual differences in latent preferences (taste) and choice consistency (scale) through the parameters $\delta$ and $\gamma$, respectively. When $\delta=0$ (so that $\lambda_j=1$ for all individuals), the GMNL model reduces to the mixed logit model. The mixed logit model captures heterogeneity in consumer preferences by allowing individuals to weight the choice attributes differently (i.e., individual differences in attribute utilities). By introducing taste heterogeneity, the mixed logit model avoids the restrictive independence of irrelevant alternatives property of the standard multinomial logit model.

The GMNL model additionally allows for differing choice consistency across individuals, known as scale heterogeneity, through the random variable $\lambda_j$. This variable changes all attribute weights simultaneously, allowing the model to predict choices to be more random for some consumers than others, such that smaller values of $\delta$ (so that $\lambda_j\rightarrow 1$) indicate differences in latent preferences where larger values of $\delta$ (so that $\lambda_j\rightarrow 0$) indicate differences in choice consistency.  The $\gamma$ parameter weights the specification between two alternative ways of introducing heterogeneity into the model.

The parameter vector is $\theta=(\beta_{0},\beta_1,\ldots,\beta_K,\sigma_{0}^2,\sigma_1^2,\ldots,\sigma_K^2,\delta^2,\gamma)'$, while the vector of random effects for each individual is $\alpha_{j}=(\eta_{0j},\ldots,\eta_{jK},\lambda_j)$.
The likelihood is therefore
\begin{align}
\label{eq:gmnllik}
p(y|\theta)=\prod_{j=1}^{S}p(y_j|\theta)=\prod_{j=1}^{S}\left[\int  \(\prod_{i=1}^{T} p(y_{ij}|\alpha_{j},\theta)\) p(\alpha_j|\theta) d\alpha_j\right],
\end{align}
where $y_{ij}$ is the observed choice, $y=(y_{11},\ldots,y_{1T},\ldots,y_{j1},\ldots,y_{ST})'$ and $p(y_{ij}|\alpha_{j},\theta)$ is given by the choice probability in Equation \eqref{eq:choiceprob}. We use diffuse priors specified as
$\beta_{0}\sim \N(0,100)$, $\sigma_{0} \propto (1+\sigma_{0}^2)^{-1}$,
$\beta_k\sim \N(0,100)$, $\sigma_k \propto (1+\sigma_{k}^2)^{-1}$, for $k=1,\ldots,K$,
$\delta\propto (1+\delta/0.2)^{-1}$,
and $\gamma \sim \textrm{U}(0,1)$.
The standard deviation parameters have half-Cauchy priors \citep{Gelman2006}.
Section \ref{Appendix:estimation-details} of the Appendix provides complete details of the estimation procedure.

\subsubsection{Results}

We estimated the posterior distribution using $M=50,000$ importance samples for the parameters, and estimated the Monte Carlo standard errors by bootstrapping the importance samples because we found that bootstrapping gave the most stable results of the estimator. 
Appendix \ref{Appendix:estimation-details} gives implementation details.
The log marginal likelihood (with standard error) for the mixed logit and GMNL models were $-981.24$ (.003) and $-978.27$ (.012), respectively. The Monte Carlo standard errors are very small, indicating the IS$^2$ method is highly efficient. The estimates of the log marginal likelihood allow us to calculate a Bayes factor of approximately 20 for the GMNL model over the mixed logit model; $BF_{\text{GMNL:mixed}}$ = $\exp(-978.27 -(-981.24))\approx 19.5$.

The larger marginal likelihood for the GMNL model than the mixed logit model provides some evidence for the presence of scale heterogeneity for this data set. Corroborating evidence comes from the parameter estimates. The posterior mean for $\gamma$ was .17 (90\% credible interval [.014--.452]), where the weight for scale heterogeneity in the GMNL model was $1-\gamma$. This means that assuming individuals only differed in their latent preferences (taste heterogeneity) did not sufficiently capture the trends in data. Some participants also made more variable choices than others (scale heterogeneity).

\subsection{The speed-accuracy tradeoff in perceptual decisions} \label{subsec:Forstmann}

Our second application explores the cognitive processes involved in perceptual decisions. We examine decision making in the context of the well-studied speed-accuracy tradeoff: the finding that electing to make faster decisions decreases decision accuracy and, conversely, that electing to increase decision accuracy causes decisions to become slower \citep[see][for a review]{heitz2014speed}. The speed-accuracy tradeoff is typically attributed to
changes in caution -- the quantity of evidence required to trigger a response (i.e., the response threshold). Here, we apply the $\textrm{IS}^{2}$ method to a class of decision making models in the Linear Ballistic Accumulator \citep[LBA;][]{Brown2008} framework to confirm that model choice via the marginal likelihood is consistent with existing methods of model choice in the literature.

We draw upon data reported by \citet{Forstmann2008} coming from an experiment in which $S=19$ participants were required to make difficult motion discrimination decisions under varying degrees of time pressure. The participants made repeated decisions about whether a cloud of semi-randomly moving dots appeared to move to the left or to the right of a display. Participants made these decisions under three conditions that differed in their task instructions: they were asked to respond as accurately as possible (condition 1: ``accuracy emphasis''), at their own pace (condition 2: ``neutral emphasis''), or as quickly as possible (condition 3: ``speed emphasis''). The three conditions were randomized across trials. A visual cue described the decision emphasis on a trial-by-trial basis prior to stimulus onset. Each subject made 280 decisions in each condition for a total of $T=840$ trials. See \citet{Forstmann2008} for all remaining details.

Freely available code applying the IS$^2$ method to one of the hierarchical LBA models described here is available at \url{osf.io/xv59c}. Section \ref{Appendix:additional-LBA-applications} of the Appendix also provides two additional applications of the IS$^2$ method to the hierarchical LBA: one application to the speed-accuracy tradeoff in lexical decisions, and a second application to response bias in lexical decisions. Code for these two examples is also available at the same repository.

\subsubsection{Behavioural model}

We analyse the choice and response time data with a hierarchical LBA model. In a typical perceptual decision making experiment, the $i^{th}$ observation
for the $j^{th}$ participant contains two pieces
of information. The first is the response choice, denoted by
$RE_{ij}\in\left\{ 1,...,C\right\} $, where $C$ is the number of
response alternatives. The second is the response time
(RT), denoted by $RT_{ij}\in\left(0,\infty\right)$. \citet{Brown2008}
derived the joint density and cumulative distribution function of
the finishing time of an LBA accumulator over response choice and
response time. The finishing time distribution for each of the LBA
accumulators is determined by five parameters: the mean and standard
deviation  of the drift rate distribution $\left(v,s\right)$, the
threshold parameter $b$, the width of the start point distribution $A$,
and the non-decision time $\tau$. It is common to set the standard
deviation $s$ of the drift rate distribution to one, i.e. $s=1$
\citep[see also][]{DonkinBrownHeathcote2009}. The probability density of the $c^{th}$
accumulator reaching the threshold at time $t$, and the other accumulators
not reaching the threshold prior to time $t$, is
\[
\textrm{LBA}\left({RE_{ij}}=c,RT_{ij}=t|b,A,v,s,\tau\right)=f_{c}\left(t\right)\prod_{k\neq c}\left(1-F_{k}\left(t\right)\right),
\]
where $f_c$ and $F_c$ are the density and distribution functions for the finishing times of a single LBA accumulator \citep[for details see][]{Brown2008,terry2015generalising}. In a two-choice experiment where response bias is not expected, it is convenient to define the means of the drift rate distributions as ${v}=\left(v^{\left(e\right)},v^{\left(c\right)}\right)$,
where $v^{\left(e\right)}$ is the mean of the drift rate for the accumulator
corresponding to the incorrect response and $v^{\left(c\right)}$ is the mean of the
drift rate for the accumulator corresponding to the correct response.
Together, these assumptions imply that each subject $j$, $j= 1,...,S$,
has the vector of random effects
\[
\left\{ b_{j},A_j,v^{\left(e\right)}_j,v^{\left(c\right)}_j,\tau_j\right\}.
\]
With the usual assumption of independence, the conditional density
of all the observations is
\[
p\left(RE,RT|{b},{A},{\tau},{v}\right)=\prod_{j=1}^{S}\prod_{i=1}^{T}\textrm{LBA}\left(RE_{ij},RT_{ij}|b_{j},A_{j},v^{\left(e\right)}_j,v^{\left(c\right)}_j,\tau_{j}\right).
\]
We follow the notation and assumptions in \citet{Gunawan2019}. For each subject $j=1,...,S$, let the vector of random effects
\[
{\alpha}_{j}=\left(\alpha_{b_{j}},\alpha_{A_{j}},\alpha_{v_{j}^{\left(e\right)}},\alpha_{v_{j}^{\left(c\right)}},\alpha_{\tau_{j}}\right),
\]
where $\alpha_{b_{j}}=\log\left(b_{j}\right)$,
$\alpha_{A_{j}}=\log\left(A_{j}\right)$, $\alpha_{v_{j}^{\left(e\right)}}=\log\left(v_{j}^{\left(e\right)}\right)$, $\alpha_{v_{j}^{\left(c\right)}}=\log\left(v_{j}^{\left(c\right)}\right)$
and $\alpha_{\tau_{j}}=\log\left(\tau_{j}\right)$. We take the following prior densities (as in \citealp{Gunawan2019}): ${\alpha}_{j}$ is $N\left({\alpha}_{j};{\mu}_{\alpha},{\Sigma}_{\alpha}\right)$, ${\mu}_{\alpha}\sim N\left(0,I_{d_{\alpha}}\right)$ where $d_{\alpha}$ is the dimension of ${\alpha}_{j}$, and the hierarchical prior for ${\Sigma}_{\alpha}$ is

\begin{eqnarray}
{\Sigma}_{\alpha}|a_{1},...,a_{d_{\alpha}} & \sim & IW\left(v_{\alpha}+d_{\alpha}-1,\ 2v_{\alpha}\textrm{diag}\left(1/a_{1},...,1/a_{d_{\alpha}}\right)\right),\\
a_{1},...,a_{d_{\alpha}} & \sim & IG\left(\frac{1}{2}, \frac{1}{\mathcal{A}_{d}^{2}}\right),\ d=1,...,d_{\alpha},
\end{eqnarray}
where $v_{\alpha}$, $\mathcal{A}_{1},...,\mathcal{A}_{d_{\alpha}}$
are positive scalars and $\textrm{diag}\left(1/a_{1},...,1/a_{d_{\alpha}}\right)$
is the diagonal matrix with diagonal elements $1/a_{1},...,1/a_{d_{\alpha}}$.\footnote{The notation $IW\left(a,A\right)$ means an inverse Wishart distribution with degrees of freedom $a$ and scale matrix $A$ and the notation $IG\left(a,b\right)$ means an inverse Gamma distribution with scale
parameter $a$ and shape parameter $b$.} This is the marginally non-informative prior of \citet{Huang2013}. Following \citet{Gunawan2019}, we set $v_{\alpha}=2$ and $\mathcal{A}_{d}=1$
for all $d$. These settings lead to half-$t$ marginal prior distributions for the standard deviations of $\Sigma_\alpha$ and uniform marginal prior distributions for the off-diagonal correlations implied by $\Sigma_\alpha$.

We tested whether the changes in task instructions across the three experimental conditions caused changes in observed decision behaviour that is consistent with changes in decision caution. We compared four LBA models. For consistency with the notation above, we refer to instructions that emphasize accuracy, neutral or speedy decisions as $a$, $n$ and $s$, respectively, and the correct and error accumulators as $c$ and $e$, respectively.

\begin{description}
\item[Model I] assumes no differences in parameters across the three conditions, i.e., a null model. This model corresponds to the psychological assumption that performance is not influenced by the instructions given to participants. The vector of random effects for subject $j$ is
\[
\alpha_{j}=\left(\alpha_{b_{j}},\alpha_{A_{j}},\alpha_{v_{j}^{\left(e\right)}},\alpha_{v_{j}^{\left(c\right)}},\alpha_{\tau_{j}}\right).
\]
\item[Model II] assumes there are two response threshold parameters. One threshold parameter defines decision-making in the accuracy- and neutral-emphasis conditions, and the second threshold parameter in the speed-emphasis condition. This makes the simplifying assumption that performance in two of the three conditions is so similar as to be effectively identically distributed,
\[
\alpha_{j}=\left(\alpha_{b_{j}^{\left(a/n\right)}},\alpha_{b_{j}^{\left(s\right)}},\alpha_{A_{j}},\alpha_{v_{j}^{\left(e\right)}},\alpha_{v_{j}^{\left(c\right)}},\alpha_{\tau_{j}}\right).
\]

\item[Model III] assumes there are three response threshold parameters, one for each condition. This model assumes that the participant responds to each condition with a different level of caution, though otherwise with identically distributed parameters,
\[
\alpha_{j}=\left(\alpha_{b_{j}^{\left(a\right)}},\alpha_{b_{j}^{\left(n\right)}},\alpha_{b_{j}^{\left(s\right)}},\alpha_{A_{j}},\alpha_{v_{j}^{\left(e\right)}},\alpha_{v_{j}^{\left(c\right)}},\alpha_{\tau_{j}}\right).
\]

\item[Model IV] assumes there are separate response threshold and non-decision time parameters for each condition. This model assumes that instructions to emphasise speedy, neutral or accurate decisions influence not only cautiousness but also the time required to perceptually encode the stimulus and produce a motor response (non-decision time),
\[
\alpha_{j}=\left(\alpha_{b_{j}^{\left(a\right)}},\alpha_{b_{j}^{\left(n\right)}},\alpha_{b_{j}^{\left(s\right)}},\alpha_{A_{j}},\alpha_{v_{j}^{\left(e\right)}},\alpha_{v_{j}^{\left(c\right)}},\alpha_{\tau_{j}^{\left(a\right)}},\alpha_{\tau_{j}^{\left(n\right)}},\alpha_{\tau_{j}^{\left(s\right)}}\right).
\]
\end{description}
As in the first application, Section \ref{Appendix:estimation-details} of the Appendix provides complete details of the estimation procedure.

\subsubsection{Results}

Once we obtained samples from the posterior, we ran the $\IS^2$ algorithm with $M=10,000$ draws to estimate the log of the marginal likelihood, $\log\ \widehat{p}(y)$, and adaptively set the number of particles $N$ that were used to estimate the likelihood such that the variance of the log-likelihood estimates did not exceed $1$; recall that $\sigma_{opt}^{2}=1$ is the optimal variance of the log-likelihood (cf. Section \ref{Sec:IS2algorithm}). We initiated this adaptive procedure by starting with $N=250$ particles and computed the variance of the log-likelihood given in Equation \eqref{varianceeqn}. If the variance was greater than 1, we increased the number of particles. We again estimated the Monte Carlo standard error of the log of the marginal likelihood estimates by bootstrapping the importance samples because we found that bootstrapping gave the most stable results.

The log marginal likelihoods (with standard errors) for Models I, II, III and IV are, respectively, 5204.17 (0.11), 7352.75 (0.06), 7453.73 (0.10) and 7521.44 (0.17), with a total computation time of around 80 minutes. As in the first application, the Monte Carlo standard errors are very small, suggesting that IS$^2$ is highly efficient. The marginal likelihoods favor Model IV -- allowing a separate response threshold and a separate non-decision time for each instruction condition. This outcome is partially consistent with \citet{Forstmann2008}, but also suggests that the non-decision time -- the combined time required to perceptually encode a stimulus and produce a motor response -- changes with instruction condition. 

\section{General Discussion and Future Work\label{sec:Conclusions}}

Model comparison is the means by which competing theories are rigorously evaluated and compared, which is fundamental to advancing psychological science. Many quantitative approaches to model comparison in psychology have struggled to appropriately account for model flexibility -- the range of data patterns that a model can predict. For example, commonly used methods such as AIC and BIC measure the flexibility of a model simply by counting the number of freely estimated model parameters. This is problematic because not all parameters are equal, in terms of complexity, and models with the same numbers of parameters can be quite different in complexity. In some cases, adding extra parameters to a model can even \emph{decrease} the complexity, for example when a hierarchical distribution structure is added to constrain otherwise-independent models for individual participants. Model selection via the marginal likelihood is one of the few methods that naturally accounts for model flexibility.

Despite its theoretical advantages, the marginal likelihood has seen limited uptake in practice for all but the simplest of psychological models, owing to its prohibitive computational cost. Although many psychologically interesting models have analytic, numerical or rapidly-simulated solutions for the responses of individual participants (e.g., the density of observed choices and response times of individual participants), the likelihood for hierarchical versions of those same models can be analytically intractable because they contain an integral over the individual random effects. This has proved to be the stumbling block in previous attempts to compute the marginal likelihood for hierarchical implementations of psychological models.

Our article develops a method for Bayesian model comparison when the likelihood (of a hierarchical model) is intractable but can be estimated unbiasedly. We show that when the density of individual observations conditioned on group- and individual-level parameters is available, an importance sampling algorithm can provide an unbiased estimate of the marginal likelihood of hierarchical models. We term this approach importance sampling squared ($\textrm{IS}^2$). The $\textrm{IS}^2$ method can be applied to samples from the posterior distribution over a model's parameters obtained following any sampling method. Section \ref{Appendix:IS2-effect-on-N} of the Appendix studies the convergence properties of the $\textrm{IS}^{2}$ estimator and provides practical guidelines for optimally selecting the number of particles required to estimate the likelihood.

We show how the $\textrm{IS}^2$ method can be used to estimate the marginal likelihood in hierarchical models of health decision making and choice response time models. In both applications, the method estimates the marginal likelihood in a principled way, providing conclusions that are consistent with the current literature. In all cases, the marginal likelihood is estimated with very small Monte-Carlo standard errors. To aid researchers in using the IS$^2$ method in their own research, we provide scripts implementing the key elements of the method, as applied to data from \citet{Forstmann2008} and \citet{Wagenmakers2008} -- see \url{osf.io/xv59c} for more details.

The data applications highlight two important properties of the $\textrm{IS}^2$ method: it is an efficient and unbiased estimator of the marginal likelihood, and it provides a standard error of the estimator. The marginal likelihood is a key quantity in appropriately accounting for model flexibility in quantitative model comparison. The standard error of an estimator is essential for interpreting any model comparison metric, not just marginal likelihoods, as it expresses the variability of the estimated metric. This is equally important in quantitatively comparing psychological models as it is in the interpretation of conventional statistics, despite being routinely overlooked. No researcher would conclude the population means of two groups differ solely on the basis of a difference in their sampled group means -- we would demand a measure that expresses the magnitude of the mean difference as a function of its variability, like a $t$-test. The $\textrm{IS}^2$ method provides the researcher with the tools to do precisely this with potentially complex and non-linear psychological models: it estimates the magnitude of differences (marginal likelihoods for different models) as a function of the variability of those estimates (the standard errors).

We emphasize that the $\textrm{IS}^2$ method is general and can be used to estimate the marginal likelihood and its standard error from models where the density of individual observations conditioned on group-level parameters and a vector of individual-level parameters is available. This is quite a general scenario -- reaching beyond the models studied here, and decision-making models more generally -- that applies to a very large category of psychological models for which hierarchical Bayesian parameter estimation has been used to date. This generality holds great promise for $\textrm{IS}^2$ as a vehicle for performing model comparison via the marginal likelihood in psychological research.

We also note that the IS$^2$ method is robust, by which we mean that the MCMC draws used to form the proposals
do not have  to converge to draws from the exact posterior as they are used in IS$^2$ to form importance sampling proposals.
Hence, as long as they are roughly in the same region as the posterior, the marginal likelihood estimates will be simulation consistent.
The same remarks hold if the MCMC targets a slightly perturbed posterior.

For faster computation, we conjecture that future work could speed up the estimation of the marginal likelihood as follows:
a) First, run the MCMC procedure in parallel on $K$ chains
without worrying about the precise convergence of each chain. In practice this means we can run each chain for far fewer iterates
than would be required for carrying out Bayesian inference if the inference was based only on the output of each chain.
b) Second,  form a proposal density based on the output of all the chains. c) Third, use IS to get $K$ robust and unbiased estimates of the marginal likelihood from each of $K$ processors. d) Finally,
average the $K$ estimates to get a final unbiased estimate of the marginal likelihood whose variance is $1/K$th the variance of each individual estimator.
The robustness and unbiasedness properties of the IS$^2$ estimator are crucial in ensuring that bias does not dominate the
variance as $K$ becomes large.

We note in Section~\ref{Appendix:IS2-method} of the Appendix that posterior expectations of functions of the parameters
can similarly be estimated with IS$^2$, albeit with some very minor bias.
We conjecture that such posterior moment estimators will be far more efficient than those obtained from standard MCMC output,
but leave to future work such speeding up of the computations of the marginal likelihood and posterior moments.

Finally, we note that the marginal likelihood is sometimes sensitive to the prior, and this might lead to some controversy in using marginal likelihood for model comparison. Addressing this issue is beyond the scope of this paper, and we focus only on how to estimate the marginal likelihood efficiently when the likelihood function is intractable.

\newpage





\pagebreak
\renewcommand{\thealgorithm}{A\arabic{algorithm}}
\renewcommand{\theremark}{A\arabic{remark}}
\renewcommand{\theequation}{A\arabic{equation}}
\renewcommand{\thetheorem}{A\arabic{theorem}}
\renewcommand{\thesection}{A\arabic{section}}
\renewcommand{\thepage}{A\arabic{page}}
\renewcommand{\thetable}{A\arabic{table}}
\renewcommand{\thefigure}{A\arabic{figure}}
\renewcommand{\theassumption}{A\arabic{assumption}}
\renewcommand{\theproposition}{A\arabic{proposition}}
\renewcommand{\thelemma}{A\arabic{lemma}}

\setcounter{page}{1}
\setcounter{section}{0}
\setcounter{equation}{0}
\setcounter{algorithm}{0}
\setcounter{table}{0}
\setcounter{figure}{0}
\setcounter{theorem}{0}
\section*{Appendix}
\section{Using the $\textrm{IS}^{2}$ Method to Obtain Posterior Expectations}\label{Appendix:IS2-method}

The motivation for using $\IS^2$ so far in this article was to develop an efficient and robust estimator for the marginal likelihood and the standard error of the estimator. This section shows that it is also straightforward to use $\IS^2$ to robustly estimate expectations with respect to the posterior distribution, as well as the standard errors of these estimators. We also discuss the convergence properties of these estimators.

Using the same notation as in Section~\ref{Sec:marginal likelihood in hierarchical models}, the posterior expectation of the function $\varphi$ of $ \theta$ is
\begin{align}\label{e:integral}
\E_\pi(\varphi) & =\int_\Theta\varphi(\t)\pi(\t)d\t =  \frac{\int_\Theta\varphi(\t)p(y|\theta)p(\theta) d\theta }
{\int_\Theta p(y|\theta)p(\theta) d\theta }.
\end{align}

When the likelihood is intractable, but can be estimated unbiasedly, we can use IS$^2$ to estimate unbiasedly and robustly
both the numerator and denominator on the right side of Equation \eqref{e:integral}, similarly to Section~\ref{Sec:IS2algorithm}, to  obtain
\beq\label{e:standardISestimator}
{\wh \E}_\pi(\varphi) = \frac{\frac1{M}\sum_{i=1}^{M}\varphi(\t_i)\wt w(\t_i)}{\frac1{M}\sum_{i=1}^{M}\wt w(\t_i)},\;\;\text{with weights}\;\;\wt w(\t_i)=\frac{p(\t_i)\wh p_N(y|\t_i)}{g_\IS(\t_i)},\;\;\t_i\stackrel{iid}{\sim}g_\IS(\t).
\eeq
In Equation \eqref{e:standardISestimator}, $\wh p_N(y|\t)$ is the unbiased estimate of $p(y|\theta)$, with $N$ the number of samples or particles used
to estimate the likelihood. It is clear that under mild conditions both the numerator and denominator of \eqref{e:standardISestimator} will converge
to the numerator and denominator respectively of \eqref{e:integral}  as $M \rightarrow \infty$, and hence ${\wh \E}_\pi(\varphi)$ converges
to $\E_\pi(\varphi)$. It is also clear that under very mild conditions the numerator of \eqref{e:standardISestimator} becomes normally distributed as
$M \rightarrow \infty$, which then means that ${\wh \E}_\pi(\varphi)$  also converges to normality. These issues are discussed more rigorously below.

We note that while both the numerator and denominator in \eqref{e:integral} are estimated unbiasedly, ${\wh \E}_\pi(\varphi)$ is a biased
estimator of $\E_\pi(\varphi)$, although the bias goes to 0 as $M \rightarrow \infty$.

In practice, the variances of both the numerator and denominator in \eqref{e:standardISestimator}, as well as their covariance, can be estimated
by the bootstrap and hence both the variance and bias of ${\wh \E}_\pi(\varphi)$ can be estimated.

\subsection{Some technical results\label{SS: technical results Sec 6}}
We now give some large sample (in $M$) convergence results for ${\wh \E}_\pi(\varphi)$ assuming that,
\begin{assumption}\label{ass: exp}  
 $\E[\wh p_N(y|\t)]=p(y|\t)$ for every $\theta\in\Theta$, where the expectation is with respect to the random variables generated
in the process of estimating the likelihood.
\end{assumption}
It is useful in the rest of this section and Section~\ref{Proofstheorem} to
follow \cite{Pitt:2012} and write $\wh p_N(y|\t)$ as $p(y|\t)e^z$,
where $z:=\log\;\wh p_N(y|\t)-\log\;p(y|\t)$ is a scalar random variable
whose distribution conditional on $\t$ is governed by the randomness due to estimating the likelihood $p(y|\t)$. Thus, the scalar $z$ 
replaces the multivariate $u$ in Section~\ref{Sec:IS2algorithm}. 
Let $g_N(z|\t)$ be the density of $z$ conditional on $\t$. Assumption~\ref{ass: exp} implies that
$\E(e^z|\theta)=\int_\mathbb{R} e^zg_N(z|\t)dz=1$.

\begin{theorem}\label{the:IS1}   
Suppose that Assumption \ref{ass: exp} holds, $\E_\pi(|\varphi|)$ exists and is finite,
and $\Sup(\pi)\subseteq\Sup(g_\text{IS})$, where $\Sup(\pi)$ denotes the support of the distribution $\pi$.
\begin{itemize}
\item[(i)] For any $N\geq 1$,  ${\wh \E}_\pi(\varphi) \stackrel{a.s.}{\longrightarrow}\E_\pi(\varphi)$ as $M\to\infty$.
\item[(ii)] If
\begin{align}\label{eq: finite var}
\int h(\theta)^2 \left (  \frac{\pi(\theta) } { g_{\IS}(\theta) } \right ) ^2  \Big ( \int  \exp(2z)g_N(z|\theta)dz \Big )  g_{\IS}(\theta)  d\theta
\end{align}
is finite for $h(\theta) =  \varphi(\theta) $ and $h(\theta) = 1 $ for all $N$, then
\beqn
\sqrt{M}\Big({\wh \E}_\pi(\varphi)-\E_\pi(\varphi)\Big)\stackrel{d}{\to}\N\big(0,\sigma^2_{\IS^2}(\varphi)\big),\; as \;M\to\infty,
\eeqn
where the asymptotic variance in $M$ for fixed $N$ is given by
\beq\label{e:ISvar_est_llh}
\sigma^2_{\IS^2}(\varphi)=\E_\pi\left\{\big(\varphi(\t)-\E_\pi(\varphi)\big)^2\frac{\pi(\t)}{g_\IS(\t)}\E_{g_N}[\exp(2z)]\right\}.
\eeq
\item[(iii)] Define
\beq\label{e:ISvar_est1}
\widehat\sigma^2_{\IS^2}(\varphi):=\frac{M\sum_{i=1}^{M}\big(\varphi(\t_i)-{\wh \E}_\pi(\varphi)\big)^2\wt w(\t_i)^2}{\left(\sum_{i=1}^{M}\wt w(\t_i)\right)^2}.
\eeq
If the conditions in (ii) hold, then $\widehat\sigma^2_{\IS^2}(\varphi)\stackrel{a.s.}{\longrightarrow}\sigma^2_{\IS^2}(\varphi)$ as $M\to\infty$, for given $N$. The proof is in Section \ref{Proofstheorem}.
\end{itemize}
\end{theorem}

Although both $\sigma^2_{\IS^2}(\varphi)$ and $\widehat\sigma^2_{\IS^2}(\varphi)$ depend on $N$,
we do not show this  dependence explicitly to simplify the notation.
Here, all the probabilistic statements, such as the almost sure convergence,
must be understood on the extended probability space that takes into account the extra randomness occurring when estimating the likelihood.
 When the likelihood can be computed, the analogous results are well known in the literature \cite[see, e.g., p.1318,][]{Geweke:1989}.


\section{The Effect on Importance Sampling of Estimating the Likelihood}\label{Appendix:IS2-effect-on-N}

The results in the previous section show that
it is straightforward to use importance sampling
even when the likelihood is intractable but unbiasedly estimated.
This section addresses the  question of how much asymptotic efficiency is lost
when working with an estimated likelihood.
We follow \cite{Pitt:2012} and \cite{Doucet2015} and
make the following idealized assumption to make it possible to develop some intuition and practical guidelines for selecting $N$. Section \ref{Proofstheorem} contains all the proofs unless stated otherwise.

\begin{assumption} \label{ass: normal}  
\begin{enumerate}
\item[(i)]
There exists a function $\gamma^2(\t)$ such that the density $g_N(z|\t)$ of $z$ is $\N(-\frac{\gamma^2(\t)}{2N},\frac{\gamma^2(\t)}{N})$, where $\N(a,b^2)$ is a univariate normal
density with mean $a$ and variance $b^2$.
\item[(ii)]
 For a given $\sigma^2>0$, define $N_{\s^2}(\t):=\gamma^2(\t)/\s^2$. Then, $\Var(z|\theta, N =N_{\s^2}(\t) )\equiv\s^2$ for all $\theta \in \Theta$.
\end{enumerate}
\end{assumption}
If $g_N(z|\theta)$ is Gaussian, then, by Assumption \ref{ass: exp}, its mean is $-\frac12 $ times its variance because $\E_{g_N}(\exp(z))=1$.
Assumption~\ref{ass: normal}~(ii) keeps the variance $\Var(z|\theta, N)$ constant across different values of $\t$, thus
making it easy to associate the $\IS^2$ asymptotic variances with $\s$.
Under Assumption~\ref{ass: normal}, the density $g_N(z|\t)$ depends only on $\s$
and we write it as $g(z|\s)$.
\begin{lemma}\label{lem:conditions}  
If Assumption \ref{ass: normal} holds for a fixed $\s^2$, then Equation \eqref{eq: finite var} becomes
\beq\label{eq: stand IS conds}
\int h(\theta)^2 \left(\frac{\pi(\t)}{g_\IS(\t)}\right)^2 g_\IS(\t)d\t<\infty 
\eeq
for both $h = \varphi$ and $h = 1$.
\end{lemma}
These are the standard conditions for IS \citep{Geweke:1989}.
The proof of this lemma is straightforward and omitted.

Recall that $\sigma^2_\IS(\varphi)$ and $\sigma^2_{\IS^2}(\varphi)$ are
respectively the asymptotic variances of the IS estimators
we would obtain when the likelihood is available
and when it is estimated.
We refer to the ratio $\sigma^2_{\IS^2}(\varphi)/\sigma^2_\IS(\varphi)$ as the inflation factor,
which measures how much the asymptotic variance is inflated when working with an estimated likelihood.
Theorem~\ref{the:ISefficency} obtains an expression for the inflation factor, shows that it is independent of $\varphi$, greater than or equal to 1 and increases exponentially with $\sigma^2$, which is the variance of $z$.

\begin{theorem}\label{the:ISefficency}  
Under Assumption~\ref{ass: normal} and the conditions in Theorem \ref{the:IS1},
\beq\label{e:IF_IS}
\frac{\sigma^2_{\IS^2}(\varphi)}{\sigma^2_\IS(\varphi)}=\exp(\sigma^2).
 \eeq
\end{theorem}

\subsection{Optimal $N$ for estimating posterior expectations}\label{Subsec:tradeoffIS}
This section aims to explicitly take into account both the statistical precision
of the $\IS^2$ estimators and
their computational cost. It is apparent that there is  a trade-off
between these two considerations. A large value of $N$  results in a precise
estimator so that $\sigma^2$ is small and the relative variance $\Var(\wh\varphi_{\IS^2})/\Var(\wh\varphi_{\IS})={\sigma^2_{\IS^2}(\varphi)}/{\sigma^2_\IS(\varphi)}=\exp(\sigma^2)$ given by
Equation \eqref{e:IF_IS} is close to 1. However, the cost of such an estimator will be large due
to the large number of particles, $N$, required. Conversely, a small value of $N$ 
results in an estimator that is cheap to evaluate but has a large value
of $\sigma^2$ and hence the variance of the IS$^2$ estimator relative to the IS estimator will be large. To explicitly
trade-off these considerations, we introduce the computational time $\CT(\sigma^2)$
 which is a product of the relative variance and the computational effort.
Minimising  $\CT(\sigma^2)$ results in an optimal value for
$\sigma^2$ and hence $N$.

Under Assumption~\ref{ass: normal}, $N=N_{\sigma^{2}}(\theta)=\gamma^{2}(\theta
)/\sigma^{2}$, so that the expected number of particles, over the draws of $\theta$,
is  ${\ov N}=\overline{{\gamma}^{2}}/\sigma^{2}$, where
$\overline {{\gamma}^{2}}:=\mathbb{E}_{g_{IS}}[\gamma^{2}(\theta)]$. This motivates the assumption that the computational cost is proportional to $1/\sigma^2$.
From Theorems \ref{the:IS1} and \ref{the:ISefficency},
the variance of the estimator $\wh\varphi_{\IS^2}$ based on $M$
importance samples from $g_\IS(\t)$
is approximated by
\beq\label{eq:Var_varphi}
\Var(\wh\varphi_{\IS^2})\approx\frac{\sigma^2_{\IS^2}(\varphi)}{M} = \frac{\sigma^2_{\IS}(\varphi)}{M}\exp(\sigma^2).
\eeq
If $\kappa^*$ is the target precision of the estimator $\wh\varphi_{\IS^2}$, then the required number of samples is $M={\sigma^2_{\IS}(\varphi)e^{\sigma^2}}/{\kappa^*}$,
and the required computational cost is proportional to
 \[\frac{\sigma^2_{\IS}(\varphi)e^{\sigma^2}}{\kappa^*}\times \frac{1}{\sigma^2}=\frac{\sigma^2_{\IS}(\varphi)}{\kappa^*}\times \frac{e^{\sigma^2}}{\sigma^2}.\]
Therefore, it makes sense to define the measure of the computing time of the IS$^2$ method (in order to obtain a given precision for the $\IS^2$ estimators)
as
\begin{align}\label{eq: CT def}
\CT(\sigma^{2}):=\frac{\exp(\sigma^{2})}{\s^2}.
\end{align}
It is straightforward to check that $\CT(\s^2)$ is minimized at
\beq\label{eq:sigma_opt}
\s^2_\text{opt}=1,
\eeq
and the optimal number of particles $N$ is such that
$\Var(z|\theta, N)=\Var(\log\;\wh p_{N}(y|\theta))=\s^2_\text{opt}$.

\subsection{Optimal $N$ for estimating the marginal likelihood}\label{SS: optimal N marg likel}
Section \ref{Subsec:tradeoffIS} derived the optimal $N$ for estimating integrals of the form in Equation \eqref{e:integral}.
We now show that the optimal $N$ is the same when the main goal is to estimate the marginal likelihood.

The $\IS^2$ marginal likelihood estimator with an estimated likelihood
is $\widehat{p}_{\IS^2}(y) = \sum_{i=1}^M {\wt w}(\t_i)/M$, with the weights
${\wt w}(\t_i)$ from Equation \eqref{eq: IS2 weights}.
We noted that $
\mathbb{E}_{g_{\IS}}[\omega(\theta)]=\mathbb{E}_{\widetilde{g}_{\IS}%
}[\widetilde{\omega}(\theta)]=p(y),
$
where $\widetilde{\omega}(\theta)=e^{z}\omega(\theta)$. We again wish to
compare the variance of the $\IS^2$ estimator to that of the $\IS$ estimator. Under Assumption \ref{ass: normal} that $z$ is Gaussian and
independent of $\theta$,
\begin{align*}
\mathbb{V}_{\widetilde{g}_{\IS}}[\widetilde{\omega}(\theta)/p(y)]  &
=\mathbb{E}_{\widetilde{g}_{\IS}}[e^{2z}\omega^{2}(\theta)/p(y)^{2}%
]-1=e^{\sigma^{2}}\mathbb{E}_{_{g_{\IS}}}[\omega^{2}(\theta)/p(y)^{2}]-1\\
&  =e^{\sigma^{2}}\left(  \mathbb{V}_{_{g_{\IS}}}[\omega(\theta
)/p(y)]+1\right)  -1.
\end{align*}
Consequently, the relative variance of the two schemes is given by%
\begin{align} \label{eq: ratio var marg lik}
\frac{\mathbb{V}_{\widetilde{g}_{\IS}}[\widehat{p}_{\IS^{2}}(y)]}{\mathbb{V}%
_{_{g_{\IS}}}[\widehat{p}_{\IS}(y)]} & =\frac{\mathbb{V}_{\widetilde{g}_{\IS}%
}[\widetilde{\omega}(\theta)/p(y)]}{\mathbb{V}_{_{g_{\IS}}}[\omega%
(\theta)/p(y)]}=\frac{e^{\sigma^{2}}(v+1)-1}{v},
\end{align}
where $v=\mathbb{V}_{_{g_{\IS}}}[\omega(\theta)/p(y)]=\mathbb{V}_{_{g_{\IS}}}[\pi(\theta)/g_\IS(\theta)]$.
Following the argument of section~\ref{Subsec:tradeoffIS}, the corresponding
computing time can be defined as
\[
\CT_\text{ML}(\sigma^{2}):=\frac{e^{\sigma^{2}}(v+1)-1}{\sigma^2},
\]
where the subscript ML indicates this is for the marginal likelihood estimator.

Let $\s^2_{\min}(v)$ minimize $\CT_\text{ML}(\sigma^{2})$ for a given $v$.
The following proposition, whose proof is obvious and omitted, summarizes some properties of $\s^2_{\min}(v)$, where
$\s^2_\text{opt}$ below is given by Equation \eqref{eq:sigma_opt}.
\begin{proposition}\label{proposition:marg_lik}  
\begin{enumerate}
\item[(i)] For any value of $v$, $\CT_\text{ML}(\s^2)$ is a convex function of $\sigma^{2}$;
therefore $\s^2_{\min}(v)$ is unique.
\item[(ii)] $\s^2_{\min}(v)$ increases as $v$ increases and
$\s^2_{\min}(v)\longrightarrow\s^2_\text{opt}=1$ in \eqref{eq:sigma_opt} as $v\longrightarrow\infty$.
\end{enumerate}
\end{proposition}

Table \ref{tab:marg_lik} illustrates some of the results of Proposition \ref{proposition:marg_lik} and shows
$\s^2_{\min}(v)$ and the ratio $\CT_\text{ML}(\sigma^{2}_{\min})/\CT_\text{ML}(\sigma^{2}_\text{opt})$
for some common values of $v$ used in practice. The
table shows that the values of $\sigma^{2}_{\min}$ and the ratio are insensitive to $v$, and
$\s^2_{\min}(v)\longrightarrow\s^2_\text{opt}=1$ as $v$
increases. From these observations, we advocate using $\s^2_\text{opt}=1$ as the optimal value of
the variance of the log likelihood estimates in estimating the marginal likelihood.
\begin{table}[h]
\centering%
\caption{Sensitivity of the computing time to $v=\Var_{g_\IS}[\pi(\theta)/g_\IS(\theta)]$
in estimating the marginal likelihood.}
\label{tab:marg_lik}%
\begin{tabular}
[c]{ccc}\hline
$v$ & $\sigma_{\min}^2(v)$ & $\CT_\text{ML}(\s^2_{\min})/\CT_\text{ML}(\s^2_\text{opt})$\\
\hline
1	&0.77	&0.97\\
5	&0.93	&0.99\\
10	&0.97	&1.00\\
100	&1.00	&1.00\\    	
$\infty$&1.00	&1.00\\
\hline
\end{tabular}
\end{table}


\section{Estimation Details for the Two Applications}\label{Appendix:estimation-details}

\subsection{Individual differences in health-related decisions}

We estimate the likelihood (Equation \eqref{eq:gmnllik} of the main text) by integrating out the random effects vector for each individual separately using different approaches for the mixed logit and GMNL models. For the mixed logit model, we combine the efficient importance sampling (EIS) method of \cite{ZR2007} with the defensive sampling approach of \cite{Hesterberg1995}. The importance density is the two component defensive mixture
\[h(\alpha_j|y_{j1},\ldots,y_{jT})=\pi h^\text{EIS}(\alpha_j|y_{j1},\ldots,y_{jT})+(1-\pi)p(\alpha_j),\]
where $h^\text{EIS}(\alpha_j|y_{j1},\ldots,y_{jT})$ is a multivariate Gaussian importance density obtained using EIS. Following
 \cite{Hesterberg1995}, including the natural sampler $p(\alpha_j)$ in the mixture ensures that the importance weights are bounded. We set the mixture weight as $\pi=0.5$. For the GMNL model, we follow \cite{fklw2010} and use the model density $p(\alpha_j)$ as an importance sampler. We implement this simpler approach for the GMNL model because the occurrence of large values of $\lambda_j$ causes the defensive mixture estimates of the log-likelihood to be right skewed in this case.

The repeated questioning of each individual (i.e., $T=32$ choice scenarios) implies that the log-likelihood estimates are sums of independent estimates $\log\hspace{1mm}  \widehat{p}(y_j|\theta)$ for each individual. To target a certain variance $\sigma^2$ for the log-likelihood estimator, we chose the number of particles for each individual ($N_j$) and parameter combination $\theta$ such that $\Var(\log\hspace{1mm} \widehat{p}(y_j|\theta))\approx \sigma^2/S$.  We implemented this scheme by using a fixed number of initial importance samples and then used the jackknife method to estimate $\gamma^2_j(\theta)$, the asymptotic variance of $\log \hspace{1mm} \widehat{p}(y_j|\theta)$, and selected $N_j(\theta)=\widehat{\gamma^2_j}(\theta)S/\sigma^2$. The preliminary number of particles is $N=20$ for the mixed logit model and $N=2,500$ for the GMNL model.

To construct efficient and reliable proposal densities for the parameters $g_\IS(\t)$, we used the ``Mixture of $t$ by Importance Sampling Weighted Expectation Maximization'' approach (MitISEM; \citealp{hoogerheide2012class}). MitISEM constructs a mixture of $t$ densities for approximating the target distribution by minimizing the Kullback--Leibler divergence between the target density and the $t$ mixture, and it can handle target distributions that have non-standard shapes such as multimodality and skewness; MitISEM effectively approximates the posterior of the two models as two component mixtures of multivariate Student's $t$ distributions.
We write $\widehat p_N(y|\t)=\widehat p_N(y|\t,u)$, with $u$ a fixed random number stream for all $\theta$.
The target distribution in the MitISEM is  $p(\t)\widehat p_N(y|\t,u)/p(y)$, which can be considered as the posterior $p(\theta|y,u)$ conditional on $y$ and the common random numbers $u$.
Our procedure is analogous to using common random numbers $u$ to obtain simulated maximum likelihood estimates of $\theta$ \citep[see, e.g.,][]{gourieroux1995statistics}, except that we obtain a histogram estimate of the ``posterior'' $p (\theta| u, y) \propto p(y|\theta, u ) p(u) $. This ``posterior''  is biased but is sufficient to obtain a good proposal density.

We implemented two standard variance reduction methods at each IS stage: stratified mixture sampling and antithetic sampling. The first consists of sampling from each component at the exact proportion of the mixture weights. For example, when estimating the likelihood for the mixed logit model we generate exactly $\pi N$ draws from the efficient importance density $h^\text{EIS}(\alpha_j|y_{j1},\ldots,y_{jT})$ and $(1-\pi) N$  draws from $p(\alpha_j)$. The antithetic sampling method consists of generating pairs of perfectly negatively correlated draws from each mixture component \citep[see, e.g.,][]{Ripley87}.

\subsection{The speed-accuracy tradeoff in perceptual decisions}\label{Appendix:proposalforalpha}

We used the Particle Metropolis within Gibbs (PMwG) sampler of \citet{Gunawan2019} to obtain $10,000$ posterior draws of ${\theta}$ and ${\alpha}_{1:S}$ for each of the four models described in the main text. We fitted a mixture of normal distributions to the samples from the posterior of $\theta$ to obtain the proposal density for the group-level parameters 
\begin{equation}
g_{\IS}\left({\theta}\right)=\sum_{k=1}^{K}w_{k}^{MIX}\phi\left({\theta}|{\mu}_{k},{\Sigma}_{k}\right),
\end{equation}
where $\phi(\mu,\Sigma)$ denotes the multivariate normal density function with mean ${\mu}$ and variance-covariance matrix ${\Sigma}$, and the $w_{k}^{MIX}$ are the component weights.
For practical purposes, we select the number of components $K$ using the Bayesian Information Criterion (BIC), and estimate the mixture of normals via Matlab's built-in function $\texttt{fitgmdist}$.

Proposal densities for the random effects ($m_{j}\left({\alpha}_{j}|{\theta},{y}_{j}\right)$) were constructed participant-by-participant. For the $j^{th}$ participant, we first fitted a normal distribution to the samples of $\left({\alpha}_{j},{\theta}\right)$, then derived the conditional distribution $g\left({\alpha}_{j}|{\theta}\right)\sim N\left({\alpha}_{j};{\mu}_{j,prop},{\Sigma}_{j,prop}\right)$ for $j=1,..,S$.
The proposal density for subject $j$ is the two-component defensive mixture
\begin{equation}
m_{j}\left({\alpha}_{j}|{\theta},{y}_{j}\right)=w_{\alpha}^{MIX}N\left({\alpha}_{j};{\mu}_{j,prop},{\Sigma}_{j,prop}\right)+\left(1-w_{\alpha}^{MIX}\right)p\left({\alpha}_{j}|{\theta}\right)\label{eq:proposalrandeffect}
\end{equation}
with $p\left({\alpha}_{j}|{\theta}\right)$ the prior density of $\alpha_j$.
The inclusion of the prior $p\left({\alpha}_{j}|{\theta}\right)$ ensures that the importance weights in Equation \eqref{eq: importance weights} are bounded \citep{Hesterberg1995}. 
 We set the mixture weight in Equation \eqref{eq:proposalrandeffect} to $w_{\alpha}^{MIX}=0.95$.


\section{Additional Applications of the Hierarchical LBA to Data}\label{Appendix:additional-LBA-applications}

This Appendix applies the hierarchical LBA model discussed in Section \ref{subsec:Forstmann} to two additional data sets.

\subsection{The speed-accuracy tradeoff in lexical decisions\label{subsec:Speed/Accuracy-Data}}

This section extends the analysis of the speed-accuracy tradeoff covered in the main text to judgments about the identity of strings of letters. Experiment 1 of \citet{Wagenmakers2008} had $17$ participants perform a basic lexical decision task that involved repeated decisions regarding whether letter strings were valid words or not (`non-words'). Prior to each block of trials, participants were given instructions to respond as quickly as possible (condition 1: speed emphasis) or as accurately as possible (condition 2: accuracy emphasis), where the instruction emphasis alternated between blocks. Concurrent to the speed-accuracy manipulation, word frequency was manipulated across four levels: words of very low frequency (vlf), low frequency (lf) and high frequency (hf), and non-words (nw). Participants completed 20 blocks with 96 trials per block for a total of 1920 lexical decisions per participant. See \citet{Wagenmakers2008} for all remaining details.

\citet{Rae2014} re-analysed \citet{Wagenmakers2008}'s data to test whether instructions manipulating the speed-accuracy tradeoff only cause changes in decision caution (response threshold parameters) or decision caution \textit{and} the speed of information processing (drift rate parameters). Through a large-scale maximum likelihood-based parameter estimation exercise, \citet{Rae2014} found evidence that emphasising the speed of decisions caused participants to lower their response threshold \textit{and} increase their drift rates for correct and incorrect responses. The latter result suggests that time pressure increased the overall drive to respond (an overall increase in both correct and incorrect drift rates) while decreasing accuracy.
We reassessed these findings, through the lens of the marginal likelihood, comparing five LBA models. For clarity in the following, we refer to instructions that emphasise accuracy and speed as $a$ and $s$, respectively, and the correct and error accumulators as $c$ and $e$, respectively.
\begin{description}
\item[Model I] assumes there are no behavioural differences across all conditions (i.e., a null model), corresponding to a single set of LBA parameters over conditions. The vector of random effects for subject $j$ is
\[
\alpha_{j}=\left(\alpha_{b_{j}},\alpha_{A_{j}},\alpha_{v_{j}^{\left(e\right)}},\alpha_{v_{j}^{\left(c\right)}},\alpha_{\tau_{j}}\right).
\]

\item[Model II] assumes that decision caution differs as a function of the task instructions (speed vs accuracy), thus allowing two response threshold parameters,
\[
\alpha_{j}=\left(\alpha_{b_{j}^{\left(s\right)}},\alpha_{b_{j}^{\left(a\right)}},\alpha_{A_{j}},\alpha_{v_{j}^{\left(e\right)}},\alpha_{v_{j}^{\left(c\right)}},\alpha_{\tau_{j}}\right).
\]

\item[Model III] assumes the speed of information processing for the correct response, but not the incorrect response, differs as a function of task instructions,
\[
\alpha_{j}=\left(\alpha_{b_{j}},\alpha_{A_{j}},\alpha_{v_{j}^{\left(e\right)}},\alpha_{v_{j}^{\left(c,s\right)}},\alpha_{v_{j}^{\left(c,a\right)}},\alpha_{\tau_{j}}\right),
\]
where $v_{j}^{\left(c,s\right)}$ and $v_{j}^{\left(c,a\right)}$ are the mean drift rates for
the accumulators corresponding to the correct response in the speed- and accuracy-emphasis conditions, respectively.

\item[Model IV] extends Model III to also allow variation in the speed of information processing for the incorrect response across task instructions,
\[
\alpha_{j}=\left(\alpha_{b_{j}},\alpha_{A_{j}},\alpha_{v_{j}^{\left(e,s\right)}},\alpha_{v_{j}^{\left(e,a\right)}},\alpha_{v_{j}^{\left(c,s\right)}},\alpha_{v_{j}^{\left(c,a\right)}},\alpha_{\tau_{j}}\right),
\]
where $v_{j}^{\left(e,s\right)}$ and $v_{j}^{\left(e,a\right)}$ are the mean drift rates for
the incorrect accumulators in the speed- and accuracy-emphasis conditions, respectively.

\item[Model V] combines Models II and IV, thus allowing the response threshold and the drift rates for correct and incorrect responses to vary as a function of task instructions,
\[
\alpha_{j}=\left(\alpha_{b_{j}^{\left(s\right)}},\alpha_{b_{j}^{\left(a\right)}},\alpha_{A_{j}},\alpha_{v_{j}^{\left(e,s\right)}},\alpha_{v_{j}^{\left(e,a\right)}},\alpha_{v_{j}^{\left(c,s\right)}},\alpha_{v_{j}^{\left(c,a\right)}},\alpha_{\tau_{j}}\right).
\]
\end{description}



Table \ref{tab:The-estimates-of marginal likelihood experiment 1}
shows the log of the marginal likelihood estimates (with standard errors in parentheses) for the five models. The Monte Carlo standard errors for
the estimates are again small for all models, suggesting that the $\textrm{IS}^{2}$ method is very efficient. Unsurprisingly, Model I, the most rigid model that allows no parameter variation across conditions, has the smallest marginal likelihood estimate. Model II, which allows response thresholds to vary with task instruction, provides a better representation of the data than Model III, which only allows the mean correct drift rates to vary with task instruction. Interestingly, Model IV, which allows the correct and incorrect drift rates to vary across speed and accuracy instructions, performs better than a model that only allows the response thresholds to differ over those conditions (i.e., Model II). Overall, the results favour Model V, which allows the response thresholds and the
correct and incorrect drift rates to vary with task instructions. This is consistent with the general conclusions of \citet{Rae2014}. Relative to instructions emphasising decision accuracy, speed-emphasis instructions cause people to lower their response thresholds and increase their rates of accumulation for both the correct and incorrect response options.

\begin{table}[H]
\caption{Log of the marginal likelihoods for Experiment 1 of \citet{Wagenmakers2008}. Bootstrapped standard errors are in parentheses. Computation time is in minutes.\label{tab:The-estimates-of marginal likelihood experiment 1}}

\centering{}%
\begin{tabular}{cccccc}
\hline
Model & I & II & III & IV & V\tabularnewline
\hline
$\log\;\ \widehat{p}\left(y\right)$ & $\underset{\left(0.20\right)}{-200.53}$ & $\underset{\left(0.06\right)}{4959.32}$ & $\underset{\left(0.15\right)}{2364.53}$ & $\underset{\left(0.10\right)}{5413.50}$ & $\underset{\left(0.16\right)}{5719.70}$\tabularnewline
PMwG (in minutes) & 54.70 & 62.76 & 55.88 & 57.05 & 62.03\tabularnewline
$\textrm{IS}^{2}$ (in minutes) & 31.88 & 34.27 & 34.44 & 33.54 & 35.49\tabularnewline
Total CPU time (in minutes) & 86.58 & 97.03 & 90.32 & 90.59 & 97.52\tabularnewline
\hline
\end{tabular}
\end{table}

\subsection{Biasing lexical decisions \label{subsec:Words/Non-Words-Proportion-Data}}

Here we analyse decisions in the presence of response bias. A common method for inducing biased decisions is to manipulate the proportion of stimuli from different response categories; for example, presenting more stimuli from category 1 than category 2 will result in a greater proportion of responses in favour of category 1 over 2. The effect of this form of bias manipulation is most commonly explained as a decision process that is primed to give a response for the more common stimulus on the basis of less evidence than would be required to give a response for the less common stimulus \citep{voss2004interpreting}. The LBA model captures this effect of response-relevant information through differences in the threshold across response options: the more common stimulus has a \emph{lower} response threshold than the less common stimulus \citep{Brown2008}, thus requiring less evidence to trigger a response. This parameter change has been confirmed experimentally \citep{forstmann2010neural}.

Here, we investigate whether changing the proportion of words and non-words induces
response biases in lexical decisions, as reflected in the parameter estimates of the LBA model. Experiment 2 of \citet{Wagenmakers2008} had $19$ participants perform a lexical decision task where the proportion of word vs non-word stimuli alternated across blocks: in a `word' block, 75\% of the stimuli were words and 25\% were non-words; in a `non-word' block the proportions switched such that 75\% of the stimuli were non-words and 25\% were words. Each block had an approximately equal amount of high, low, and very low frequency words. Participants completed $20$ blocks with $96$ trials per block for a total of 1920 lexical decisions per participant. See \citet{Wagenmakers2008} for all other details.

We test whether the word vs non-word response proportion manipulation affected the response threshold parameters of the LBA model in the expected direction. We used the marginal likelihood to compare five LBA models. In addition to the notation introduced above, we refer to the stimulus manipulation of the proportion of words and non-words as $w$ and $nw$, where $w$ refers to the condition with 75\% words and 25\% non-words and vice versa for $nw$, and $W$ and $NW$ to refer to a response of word and non-word, respectively.
\begin{description}
\item[Model I] again assumes constancy in LBA parameters across conditions,
\[
\alpha_{j}=\left(\alpha_{b_{j}},\alpha_{A_{j}},\alpha_{v_{j}^{\left(e\right)}},\alpha_{v_{j}^{\left(c\right)}},\alpha_{\tau_{j}}\right).
\]
\item[Model II] assumes that the response threshold varies as a function of the proportion manipulation ($w$, $nw$) but does not differ for responses of words or non-words ($W$, $NW$),
\[
\alpha_{j}=\left(\alpha_{b_{j}^{\left(w\right)}},\alpha_{b_{j}^{\left(nw\right)}},\alpha_{A_{j}},\alpha_{v_{j}^{\left(e\right)}},\alpha_{v_{j}^{\left(c\right)}},\alpha_{\tau_{j}}\right).
\]

\item[Model III] assumes that the two responses (``word'', $W$; ``non-word'', $NW$) are governed by accumulators that had different response thresholds, and these thresholds are also free to vary across the proportion manipulation,
\[
\alpha_{j}=\left(\alpha_{b_{j}^{\left(w,W\right)}},\alpha_{b_{j}^{\left(w,NW\right)}},\alpha_{b_{j}^{\left(nw,W\right)}},\alpha_{b_{j}^{\left(nw,NW\right)}},\alpha_{A_{j}},\alpha_{v_{j}^{\left(e\right)}},\alpha_{v_{j}^{\left(c\right)}},\alpha_{\tau_{j}}\right),
\]

where $b_{j}^{\left(w,W\right)}$ and $b_{j}^{\left(w,NW\right)}$ refer to the threshold parameters for word and non-word responses in the 75\% word condition, respectively, and similarly for $b_{j}^{\left(nw,W\right)}$ and $b_{j}^{\left(nw,NW\right)}$ in the 75\% non-word condition.

\item[Model IV] has a parallel structure to Model III except that the proportion manipulation ($w$, $nw$) and response ($W$, $NW$) are assumed to selectively influence non-decision time rather than the response threshold,
\[
\alpha_{j}=\left(\alpha_{b_{j}},\alpha_{A_{j}},\alpha_{v_{j}^{\left(e\right)}},\alpha_{v_{j}^{\left(c\right)}},\alpha_{\tau_{j}^{\left(w,W\right)}},\alpha_{\tau_{j}^{\left(w,NW\right)}},\alpha_{\tau_{j}^{\left(nw,W\right)}},\alpha_{\tau_{j}^{\left(nw,NW\right)}}\right),
\]
where the random effects $\tau_{j}^{\left(w,W\right)}$, $\tau_{j}^{\left(w,NW\right)}$,
$\tau_{j}^{\left(nw,W\right)}$, and $\tau_{j}^{\left(nw,NW\right)}$ were
similarly defined as in the threshold parameters of Model III.

\item[Model V] assumes that the proportion manipulation influences the rate of evidence accumulation for correct and incorrect ``word'' and ``non-word'' responses, separately for each proportion condition,
\begin{align}
\alpha_{j}=\Big(\alpha_{b_{j}},\alpha_{A_{j}},& \alpha_{v_{j}^{\left(e,w,W\right)}},\alpha_{v_{j}^{\left(e,w,NW\right)}},\alpha_{v_{j}^{\left(e,nw,W\right)}},\alpha_{v_{j}^{\left(e,nw,NW\right)}}, \nonumber \\
& \alpha_{v_{j}^{\left(c,w,W\right)}},\alpha_{v_{j}^{\left(c,w,NW\right)}},\alpha_{v_{j}^{\left(c,nw,W\right)}},\alpha_{v_{j}^{\left(c,nw,NW\right)}},\alpha_{\tau_{j}} \Big), \nonumber
\end{align}
where $v_{j}^{\left(e,\ ,\ \right)}$ and $v_{j}^{\left(c,\ ,\ \right)}$ are the vectors of mean drift rates for incorrect and correct responses, respectively. The relevant pair of correct-incorrect drift rates for a given datum are contingent on the proportion condition of the trial (75\% words or non-words; $w$, $nw$) and the observed response (word or non-word; $W$, $NW$). Phrased differently, each datum (RT and response choice pair from a single trial) contains information to update 2 of the 8 drift rate parameters.

\end{description}

Table \ref{tab:The-estimates-of marginal likelihood experiment2} shows the log of the marginal likelihoods for the five models. As in the previous applications, the Monte Carlo standard errors
for the estimates are small for all models, although the more complex model V has slightly
higher standard errors; nevertheless, even this larger standard error is still much smaller than the differences in the log of the marginal likelihoods between models.

\begin{table}[H]
\caption{Log of the marginal likelihoods for Experiment 2 of \citet{Wagenmakers2008}. Bootstrapped standard errors are in parentheses. Computation time is in minutes.\label{tab:The-estimates-of marginal likelihood experiment2}}

\centering{}%
\begin{tabular}{cccccc}
\hline
Model & I & II & III & IV & V\tabularnewline
\hline
$\log\;\ \widehat{p}\left(y\right)$ & $\underset{\left(0.09\right)}{5769.33}$ & $\underset{\left(0.05\right)}{5876.11}$ & $\underset{\left(0.11\right)}{8449.08}$ & $\underset{\left(0.13\right)}{8259.95}$ & $\underset{\left(0.39\right)}{7501.00}$\tabularnewline
PMwG (in minutes) & 100.52 & 103.01 & 107.60 & 105.40 & 116.25\tabularnewline
$\textrm{IS}^{2}$ (in minutes) & 35.82 & 39.58 & 42.96 & 45.24 & 57.88\tabularnewline
Total CPU time (in minutes) & 136.34 & 142.59 & 150.56 & 150.64 & 174.13\tabularnewline
\hline
\end{tabular}
\end{table}

Models that do not allow any parameter variation (I) or a simple threshold change across proportion condition but not response type (II) provide a poor representation of the data; these models do not account for response bias, which suggests there were considerable bias effects present in the data. Assuming independent drift rates for each combination of correct/incorrect, response type and proportion condition (V) provided a much improved fit, though it was still inferior to models that assume a selective influence of the response threshold (III) or non-decision time (IV) on the response type $\times$ proportion condition interaction. The best account of the data is one in which the biased responses are assumed to arise from a change in the response thresholds (III). The parameter estimates of the preferred model suggest that the biased responses arose from lowered response thresholds for more probable responses.


\section{Proofs}\label{Proofstheorem}
Using the same notation as in Section~\ref{SS: technical results Sec 6}, we follow \cite{Pitt:2012}
and define the joint prior density of $\theta$ and $z$ as $p(\theta)\exp(z) g_N(z|\theta)$, so that the posterior density of $\theta$ and $z$ is
\begin{align} \label{eq: post theta and z}
\wt \pi (\theta, z) & = \frac{p(y|\theta) p(\theta)\exp(z) g_N(z|\theta)}{ p(y)}
= \frac{\wh p_N(y|\theta) p(\theta) g_N(z|\theta) }{p(y)} ,
\end{align}
and has $p(\theta|y)$ as a marginal. Hence,
\begin{align}\label{eq: marg lik  in z and theta}
p(y) & = \int_\Theta \int_Z \wh  p_N(y|\theta) p(\theta) g_N(z|\theta)  dz d\theta 
 = \int_\Theta \int_Z \wt w(\theta, z) \wt g_{IS} (\theta, z) dz d\theta\\
\intertext{where}
\wt w(\theta,z) & = \frac{\wh p_N(y|\theta) p(\theta) g_N(z|\theta)} { \wt g_{IS}(\theta,z)} \quad \text{and} \quad
\wt g_{IS}(\theta,z) = p(\theta) g_N(z|\theta)
\end{align}
Thus, we can  consider $\wt g_{IS} $ in Equation \eqref{eq: marg lik  in z and theta} as an importance density in $z$ and $\theta$, which will lead to
the same IS$^2$ estimate for the marginal likelihood as in Section~\ref{Sec:IS2algorithm}.

\begin{proof}[Proof of Theorem \ref{lem: robust}]
Unbiasedness holds because we are dealing with IS. 
The assumptions in Theorem \ref{the:IS1} ensure that $\wt w(\t_i)$'s are i.i.d with a finite second moment.
The results of (i) and (ii) then follow.
\end{proof}

\begin{proof}[Proof of Theorem \ref{the:IS1}]
If $\Sup(\pi)\subseteq\Sup(g_\text{IS})$ then $\Sup(\pi_N)\subseteq\Sup(\wt g_\text{IS})$.
This, together with the existence and finiteness of $\E_\pi(\varphi)$ ensures that $\E_{\wt g_\IS}[\varphi(\t)\wt w(\t,z)]=p(y)\E_\pi(\varphi)$ and $\E_{\wt g_\IS}[\wt w(\t,z)]=p(y)$ exist and are finite.
Result (i) follows from the strong law of large numbers.
To prove (ii), write
\beqn
\wh E_\pi(\varphi)-\E_\pi(\varphi) = \frac{\frac1M\sum_{i=1}^M\big(\varphi(\t_i)-\E_\pi(\varphi)\big)\wt w(\t_i,z_i)} {\frac1{M}\sum_{i=1}^{M}\wt w(\t_i,z_i)}.
\eeqn
Let $X_i=\big(\varphi(\t_i)-\E_\pi(\varphi)\big)\wt w(\t_i,z_i)$, $i=1,...,M$, $S_{M}=\frac{1}{M}\sum_{i=1}^{M} X_i$ and $Y_{M}=\frac1{M}\sum_{i=1}^{M}\wt w(\t_i,z_i)$.
The $X_i$ are independently and identically distributed with $\E_{\wt g_\IS}(X_i)=0$ and
\bean
\nu^2_N=\Var_{\wt g_\IS}(X_i)=\E_{\wt g_\IS}(X_i^2)&=&\int_{\wt\Theta}\int_{Z}\big(\varphi(\t)-\E_\pi(\varphi)\big)^2\wt w(\t,z)p(y)\pi_N(\t,z)d\t dz\\
&=&p(y)\E_{\pi_N}\left\{\big(\varphi(\t)-\E_\pi(\varphi)\big)^2\wt w(\t,z)\right\}<\infty.
\eean
By the central limit theorem for a sum of independently and identically distributed random variables
with a finite second moment, $\sqrt{M}S_{M}\stackrel{d}{\to}\N(0,\nu^2_N)$.
By the strong law of large numbers, $Y_{M}\stackrel{P}{\to}\E_{\wt g_\IS}(\wt w(\t,z))=p(y)$.
By Slutsky's theorem,
\beqn
\sqrt{M}\Big(\wh E_\pi(\varphi)-\E_\pi(\varphi)\Big)=\frac{\sqrt{M}S_{M}}{Y_{M}}\stackrel{d}{\to}\N(0,\nu^2_N/p(y)^2).
\eeqn
The asymptotic variance is given by
\bean
\wh E_\pi(\varphi)=\frac{\nu^2_N}{p(y)^2}&=&\frac{1}{p(y)}\E_{\pi_N}\left\{\big(\varphi(\t)-\E_\pi(\varphi)\big)^2\wt w(\t,z)\right\}\\
&=&\E_\pi\left\{\big(\varphi(\t)-\E_\pi(\varphi)\big)^2\frac{\pi(\t)}{g_\IS(\t)}\E_{g_N(z|\t)}[\exp(2z)]\right\}.
\eean
To prove (iii), write
\bean
\widehat\sigma^2_{\IS^2}(\varphi)&=&\frac{\frac{1}{M}\sum_{i=1}^{M}\big(\varphi(\t_i)-  \wh E_\pi(\varphi) \big)^2\wt w(\t_i,z_i)^2}{\left(\frac{1}{M}\sum_{i=1}^{M}\wt w(\t_i,z_i)\right)^2}\\
&\stackrel{a.s.}{\to}&\frac{\E_{\wt g_\IS}\left\{\big(\varphi(\t)-\E_\pi(\varphi)\big)^2\wt w(\t,z)^2\right\}}{\Big(\E_{\wt g_\IS}(\wt w(\t,z))\Big)^2}\\
&=&\sigma^2_{\IS^2}(\varphi).
\eean
\end{proof}

\begin{proof}[Proof of Theorem \ref{the:ISefficency}]
Under Assumption 2, $g_N(z|\t)=\N(-\s^2/2,\s^2)$ and
$\E_{g_N(z|\t)}[\exp(2z)]=\exp(\s^2)$.
From Equations \eqref{e:ISvar_est_llh} and \eqref{e:ISvar_est1},
\beqn
\sigma^2_{\IS^2}(\varphi)=\E_\pi\left\{\big(\varphi(\t)-\E_\pi(\varphi)\big)^2\frac{\pi(\t)}{g_\IS(\t)}\exp(\sigma^2)\right\}=\exp(\sigma^2)\sigma^2_\IS(\varphi).
\eeqn 
\end{proof}

\newpage

\section*{Open Practices Statement}

The two applications cover previously published data sets \citep{fklw2010,Forstmann2008}. Data and code are available at \url{osf.io/xv59c}.

\newpage

\bibliographystyle{apalike}
\bibliography{references_v1}

\end{document}